\documentclass[letterpaper]{article} 
\usepackage{aaai2026}  
\usepackage{times}  
\usepackage{helvet}  
\usepackage{courier}  
\usepackage[hyphens]{url}  
\usepackage{graphicx} 
\usepackage{natbib}  
\usepackage{caption} 
\usepackage{algorithm}
\usepackage{algorithmic}

\usepackage{amsmath}
\usepackage{multirow}
\usepackage{amsthm}
\usepackage{amsfonts}
\newtheorem{proposition}{Proposition}
\newtheorem{lemma}{Lemma}
\newtheorem{assumption}{Assumption}
\newtheorem{theorem}{Theorem}

\usepackage{xcolor}
\newcommand{\R}{\mathbb{R}}
\renewcommand{\P}{\mathbb{P}}

\newcommand{\E}{\mathbb{E}}

\newcommand{\F}{\mathcal{F}}
\newcommand{\tr}{{{\mathsf T}}}

\usepackage{xcolor}
\usepackage[colorlinks]{hyperref}
\hypersetup{
    citecolor=black
}

\usepackage{newfloat}
\usepackage{listings}
\DeclareCaptionStyle{ruled}{labelfont=normalfont,labelsep=colon,strut=off} 
\floatstyle{ruled}
\newfloat{listing}{tb}{lst}{}
\floatname{listing}{Listing}
\title{DiffOP: Reinforcement Learning of Optimization-Based Control Policies \\ via Implicit Policy Gradients}

\author {
    Yuexin Bian, 
    Jie Feng, 
    Yuanyuan Shi 
}
\affiliations {
   University of California, San Diego \\
   yubian@ucsd.edu,
   jif005@ucsd.edu,
   yyshi@ucsd.edu
}

\usepackage{bibentry}

\begin{document}

\maketitle

\begin{abstract}
Real-world control systems require policies that are not only high-performing but also interpretable and robust. A promising direction toward this goal is model-based control, which learns system dynamics and cost functions from historical data and then uses these models to inform decision-making. Building on this paradigm, we introduce DiffOP, a novel framework for learning optimization-based control policies defined implicitly through optimization control problems.
Without relying on value function approximation, DiffOP jointly learns the cost and dynamics models and directly optimizes the actual control costs using policy gradients.
To enable this, we derive analytical policy gradients by applying implicit differentiation to the underlying optimization problem and integrating it with the standard policy gradient framework.
Under standard regularity conditions, we establish that DiffOP converges to an $\epsilon$-stationary point within $\mathcal{O}(\epsilon^{-1})$ iterations.
We demonstrate the effectiveness of DiffOP through experiments on nonlinear control tasks and power system voltage control with constraints.
The code is available at \href{https://github.com/alwaysbyx/DiffOP}{https://github.com/alwaysbyx/DiffOP}.
\end{abstract}

%

\section{Introduction}
\label{sec:introduction}
Operating and controlling complex physical systems in an effective manner is of critical importance to society.
Real-world physical systems, such as power grids~\cite{machowski2020power}, commercial and industrial infrastructures~\cite{jang2024active}, transportation networks~\cite{negenborn2008multi}, and robotic systems~\cite{spong2020robot}, require control policies that are not only high-performing but also interpretable with reliability and safety guarantees. 
To this end, optimization-based policies such as model predictive control (MPC) have been explored with known~\cite{morari1999model,grune2017nonlinear} or learned system dynamics~\cite{chen2019optimal}, offering a principled framework for performance optimization, constraint satisfaction, and transparent decision-making.

Optimization-based policies formulate control as a mathematical optimization problem, where the goal is to minimize a system cost subject to dynamic equations and state-action constraints.
Prior work in this area often focuses on learning system dynamics and cost functions separately from historical data to improve predictive accuracy~\cite{chen2019optimal, salzmann2023real, zhang2021trajectory}. However, this decoupled learning procedure can lead to the objective mismatch problem~\cite{donti2017task, lambert2020objective}, where the learned models perform well in isolation (e.g., accurate prediction), but fail to guide optimal control decisions—since the training objectives do not directly reflect control performance.

To enable end-to-end learning of optimization-based policies, recent works~\cite{srinivas2018universal, amos2018differentiable, jin2020pontryagin, xu2024revisiting} have made the optimization layers differentiable, allowing both the dynamics and cost function to be jointly learned. These approaches primarily focus on supervised imitation learning settings, where the optimization policy is trained to mimic expert demonstrations and achieve effective control.

More recently, research has moved beyond imitation learning to explore reinforcement learning (RL)-based training of optimization-based policies, aiming to improve the control performance through direct interaction with the environment.
In particular, existing works have investigated the integration of MPC with RL~\cite{ esfahani2021approximate, lin2023reinforcement, reiter2025synthesis,lawrence2025view}, leveraging the fact that MPC naturally defines a state-action value function $Q$, which can be updated via Q-learning to enhance closed-loop control performance.

Building on this line of work, our approach departs from value function approximation and instead treats the optimization-based policy as a fully differentiable control module. 
This work introduces DiffOP, a novel optimal control framework based on a \textbf{Diff}erentiable \textbf{O}ptimization-based \textbf{P}olicy.
DiffOP computes analytical policy gradients by integrating trajectory derivatives derived from Pontryagin’s Maximum Principle (PMP)~\cite{jin2020pontryagin, xu2024revisiting} with the standard policy gradient algorithm~\cite{sutton2018reinforcement}, enabling end-to-end learning in model-free RL settings.

Our main contributions are as follows:
\begin{itemize}
    \item We propose DiffOP, which learns optimization-based control policies via policy gradients. A key technical innovation is the joint learning of both the cost function and system dynamics with environmental rewards, using implicit differentiation and policy gradient updates.
    
    \item We provide, to the best of our knowledge, the first non-asymptotic convergence guarantee for learning optimization-based policies via policy gradients. Specifically, we show that DiffOP converges to an $\epsilon$-stationary point within $\mathcal{O}(\epsilon^{-1})$ iterations.
    
    \item We validate DiffOP on a set of nonlinear control tasks and a real-world voltage control problem with safety constraints, demonstrating its ability to learn cost and dynamics models with superior control performance compared to existing RL-based MPC methods and differentiable optimization approaches in supervised settings.
\end{itemize}

\section{Related Work}

\textbf{Differentiable Optimization for Control. } Recent advances have enabled end-to-end learning of model-based control policies (e.g., MPC) by making the underlying optimization differentiable. This allows computing gradients w.r.t. optimization parameters, enabling learning via gradient-based methods.
One approach is to unroll the optimization steps within the computational graph and apply automatic differentiation, as proposed in~\cite{srinivas2018universal,liu2018proximal}. However, this method is memory-intensive and may yield inaccurate gradients due to truncation.
To address the limitations of unrolling, \cite{amos2018differentiable} proposed differentiable MPC by approximating the original problem with iterative LQR and differentiating through its KKT conditions, though this approach does not support state constraints. Subsequently, \cite{jin2020pontryagin} derived exact analytical gradients by differentiating through Pontryagin’s Maximum Principle (PMP), enabling support for state-action constraints and improving computational efficiency. More recently, \cite{xu2024revisiting} introduced IDOC, which applies variable elimination to the KKT system to directly compute trajectory derivatives. While these methods have primarily been applied to imitation learning, our work builds on~\cite{xu2024revisiting} to enable gradient computation for direct policy optimization.


\noindent
\textbf{Synthesis of MPC and RL. } 
Supervised learning of MPC policies often relies on minimizing proxy losses (e.g., prediction error), which may not align with the true control objectives.
To overcome this, there is increasing interest in integrating RL with MPC to enhance closed-loop performance.  
For example, \cite{chen2019gnu} leverages Proximal Policy Optimization (PPO) to adapt linear dynamics models within the differentiable-MPC framework~\cite{amos2018differentiable}, relying on LQR approximations.
Given that MPC naturally defines a state-action value function, \cite{gros2019data, gros2021reinforcement, lawrence2025view} employ Q-learning with temporal-difference updates to learn improved MPC policies. Additionally,~\cite{seel2023combining} combines Q-learning with deterministic policy gradients to improve convergence. 
Our work builds on this direction but differs in that, instead of relying on value function approximation or LQR-based structures, we formulate the control policy as the solution to a parameterized nonlinear optimization problem and derive exact analytical gradients via implicit differentiation.
Beyond algorithmic design, our work contributes theoretical convergence results for the proposed framework, building on policy gradient theory~\cite{papini2018stochastic, yang2021sample, zhang2021sample} and recent work in bi-level optimization~\cite{ji2021bilevel, kwon2023fully}.

\begin{figure*}[t]
    \centering
    \includegraphics[width=1.\linewidth]{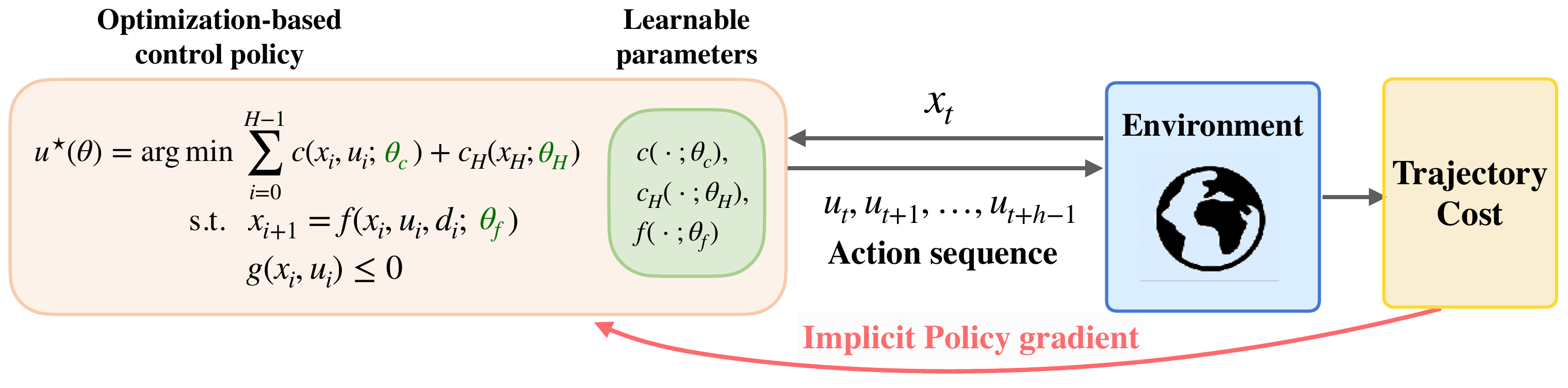}
    \caption{Overview of the DiffOP framework. An optimization-based control policy generates an action sequence by solving a parameterized optimal control problem with learnable dynamics and cost models. The environment executes the action sequence and returns cost feedback, which is used to compute policy gradients for updating the parameters.}
    \label{fig:frame}
\end{figure*}

\section{DiffOP: A Differentiable Optimization Policy}\label{sec:diffop}
In this section, we introduce DiffOP, the optimization-based control policy, and formalize its learning process as a bi-level optimization problem.
Throughout this paper, we use \(i\) to denote the planning step within the optimization horizon, and \(t\) to denote the actual time step in the control system.

\subsection{Optimization-based Control Policy}
The optimization-based control policy is defined as:
\begin{equation}
\label{eq:op_policy}
\begin{aligned}
    u_{0:H-1}^\star(x_{\text{init}}; \theta) = \arg\min_{u} \,\, & \sum_{i=0}^{H-1} c(x_{i}, u_{i}; \theta_c) + c_H(x_H; \theta_H) \\
    \text{subject to} \,\, & x_0 =x_{\text{init}}, \\
    & x_{i+1} = f(x_i, u_i; \theta_f), \\
    & g(x_i, u_i) \leq 0, {\small i = 0, ..., H-1}
\end{aligned}
\end{equation}
where $H$ is the planning horizon, $x_{\text{init}} \in \R^n$ is the initial state. We denote $x_i \in \R^n, u_i \in \R^m$ as the system state and control action at the $i$-th planning step, respectively. $c(x_{i}, u_{i}; \theta_c), c_H(x_H; \theta_H)$ model the instantaneous and terminal costs with learnable parameter $\theta_c$ and $\theta_H$, $f(x_i, u_i; \theta_f)$ models system dynamics with learnable parameters $\theta_f$, and $g(\cdot)$ represents the state and control constraints, which are known to the control agent.  
In summary, our policy~\eqref{eq:op_policy} is represented by some parametric representation of the \emph{unknown} cost and dynamics functions with parameters $\theta = (\theta_c, \theta_H, \theta_f) \in \R^d$.

Given a state $x$, DiffOP generates a sequence of control actions by solving~\eqref{eq:op_policy}. We define the policy output as:
\begin{equation}\label{eq:op_policy2}
\pi_\theta(x) = \{u_0^\star(x; \theta), \ldots, u_{h-1}^\star(x; \theta)\},
\end{equation}
where $u_i^\star$ denotes the $i$-th action from the optimized sequence, and $1 \leq h \leq H$ specifies how many actions are executed before receiving the next observation. 

\vspace{6pt}
\noindent
\textbf{Remark 1.} Although our policy optimizes over a finite horizon, we do not name it as ``MPC'' since it is not limited to receding-horizon execution. DiffOP allows execution of multiple actions per optimization step (\(h \geq 1\)). As shown in both the algorithm and experiments, DiffOP naturally accommodates arbitrary \(h\) through its policy gradient formulation, supporting both step-wise (when $h=1$ as in standard MPC) and trajectory-level control (when $h>1$). 
Moreover, any standard optimization solver~\cite{andersson2019casadi, wachter2006implementation} can be used to solve the policy.


\subsection{Policy Learning as a Bi-Level Optimization}
We are interested in learning the parameters in the policy $\theta = (\theta_c, \theta_H, \theta_f) \in \R^d$ to minimize the overall control costs. 
The policy optimization problem is formulated as:
\begin{subequations}\label{eq:problem}
\begin{align}
\min_{\theta} \quad & C(\theta) := \E\left[\sum_{t=0}^T c(x_t, u_t; \phi_c)\right] \label{eq:cost} \\
\text{subject to} \quad
& x_{t+1} = f(x_t, u_t; \phi_f), \,\, \forall t, \label{eq:dynamics} \\
& {u_t, \ldots, u_{t+h-1}} \sim \pi_{\theta}(x_t). \, \label{eq:sto}
\end{align}
\end{subequations}

\noindent \textbf{Upper-level: Unknown system~\eqref{eq:cost}\eqref{eq:dynamics}:} Equation~\eqref{eq:cost} describes the ground-truth system cost model with parameters $\phi_c$ and equation~\eqref{eq:dynamics} describes ground-truth system dynamics parameterized by $\phi_f$, both are \emph{unknown} to the agent. The decision maker seeks to learn an optimization-based policy in the form of \eqref{eq:op_policy} with parameters $\theta = (\theta_c, \theta_H, \theta_f)$ that minimizes the expected accumulated costs in the real system.

\noindent \textbf{Lower-level: Stochastic optimization-based policy~\eqref{eq:sto}:} The lower-level optimization involves solving the optimization policy to determine the control action at each time step. To enable effective exploration during policy learning, we introduce stochasticity by reparameterizing the policy ${u_t, \ldots, u_{t+h-1}} \sim \pi_{\theta}(x_t)$ as a truncated Gaussian distribution. Specifically, at each decision point $t$, the policy generates a sequence of $h$ actions by solving the optimization, and each action is sampled as: 
\begin{equation}\label{eq:sample}
u_i \sim \mathcal{N}(u_i^\star, \sigma^2 I), \,\, \text{clipped to } [u_i^\star \pm \beta \sigma^2 I],
\end{equation}
for $i = t, \ldots, t+h-1$. Here $u_i^\star$ denotes the deterministic action at time step obtained by solving the optimization-based policy~\eqref{eq:op_policy2}.  The term $\sigma^2 I$ is the covariance matrix, with $\sigma^2 > 0$ controlling the variance of the exploration noise, and $I \in \R^{m \times m}$ is the identity matrix. 
The truncation bounds $u_\tau^\star \pm \beta \sigma^2I$ ensure the sampled actions remain within a controlled range, with $\beta > 0$ serving as a hyperparameter that regulates the extent of exploration. This independent Gaussian sampling is used for simplicity and analytical tractability in convergence analysis. In practice, DiffOP can be extended to sample the action sequence jointly from a multivariate Gaussian, allowing for correlated noise across the planning horizon. The expectation $\E$ is taken over the initial state distribution $x_0 \sim \mathcal{D}$ and the trajectory $(x_0, u_0, \ldots, x_T, u_T)$ induced by the stochastic policy~\eqref{eq:sto}. 
By choosing a truncated Gaussian distribution over a standard Gaussian distribution, we explicitly bound the deviation of sampled actions from their optimal counterparts. This bounded exploration helps control the variance of the policy gradients, which is used for establishing gradient \emph{boundedness} and \emph{convergence} analysis of our algorithm. 

Crucially, we view problem~\eqref{eq:problem} as a \textit{policy optimization} task, where the policy is implicitly defined through an optimization problem. This formulation enables the adoption of policy gradient methods for end-to-end training.

\section{Proposed Algorithm}\label{sec:algorithm}
The proposed algorithm is summarized in Algorithm~\ref{alg:diffop}.
Our algorithm builds on the classic policy gradient framework~\cite{sutton2018reinforcement}, and incorporates trajectory derivatives derived from Pontryagin’s Maximum Principle (PMP)~\cite{jin2020pontryagin, xu2024revisiting} to compute gradients of the actual cost with respect to the policy parameters \(\theta := \{\theta_c, \theta_H, \theta_f\}\), which include the stage cost parameters \(\theta_c\), terminal cost parameters \(\theta_H\), and dynamics model parameters \(\theta_f\).

\begin{algorithm}[h]
\caption{DiffOP}
\label{alg:diffop}
\begin{algorithmic}
    \STATE {\bfseries Input:} $\theta^{(0)} = (\theta_c^{(0)}, \theta_H^{(0)}, \theta_f^{(0)})$, learning rate $\eta$

\FOR{$k = 0, 1, \ldots, K-1$}
    \FOR{$n = 1, \ldots, N$}
    \STATE Sample the action as in eq~\eqref{eq:op_policy2}\eqref{eq:sample}, $\forall t$
    \STATE Collect the trajectory data $(x_t^{(n)}, u_t^{(n)}, {u_t^\star}^{(n)}), \forall t$ 
     \STATE Calculate the gradient $\nabla_\theta {u_t^\star}^{(n)}$ as in \underline{Proposition 1} 
    \ENDFOR
    \STATE Estimate policy gradient $\widehat{\nabla}_\theta C(\theta^{(k)})$ as in \underline{Proposition 2} 
    \STATE Update policy $\theta^{(k+1)} \leftarrow \theta^{(k)} - \eta \widehat{\nabla}_\theta C(\theta^{(k)})$
\ENDFOR
\end{algorithmic}
\end{algorithm}

At each training iteration $k$, the algorithm collects a batch of $N$ trajectories by sampling actions from the current optimization-based policy. For each trajectory, it computes the gradients $\nabla_\theta u^\star$ of the optimal actions with respect to the policy parameters. These gradients are then used to estimate the overall policy gradient $\widehat{\nabla}_\theta C(\theta)$, which guides the parameter update via standard gradient descent. To derive these gradients, we first present Proposition~\ref{prop:grad_policy}, which characterizes the differentiability of the optimization-based policy w.r.t. its parameters using the first-order optimality conditions and the implicit function theorem.

We begin by defining $x_i^\star, u_i^\star$ as the optimal state and action obtained by solving the optimization policy~\eqref{eq:op_policy}, where $i=0, \ldots, H-1$ denotes the $i$-th planning step. We further define $\zeta^\star = (x_0^\star, u_0^\star, x_1^\star, u_1^\star, \ldots, u_{H-1}^\star, x_{H}^\star)$ as the optimal solution trajectory over the horizon $H$. 
We stack all constraints into a single vector-valued constraint $\kappa:= (\kappa_{-1}, \kappa_{0}, \kappa_{1}, \ldots, \kappa_{H}) \in \R^{n_\kappa}, n_\kappa = (H+1)n + q$. Specifically, $\kappa_{-1} = x_0^\star$ is the input to the policy, and 
\begin{equation}
    \kappa_i = 
    \begin{cases}
        (\Tilde{g}(x^\star_i, u^\star_i), x^\star_{i+1}-f(x^\star_i, u^\star_i;\theta_f)),  i = 0, \ldots, H-1, \\
        \Tilde{g}(x^\star_H), i = H,
    \end{cases}
\end{equation}
where $\Tilde{g}(x_i^\star, u_i^\star)$ are the subset of active inequality constraints at $i$-th step, and $q$ is the total number of active inequality constraints. The optimization problem in~\eqref{eq:op_policy} can be solved using general-purpose solvers~\cite{gill2005snopt,diamond2016cvxpy} to determine both the optimal solution $\zeta^\star$ and its active constraints $\kappa$.

\begin{proposition}[Gradient of the optimization-based policy]\label{prop:grad_policy}
Suppose $u^\star$ is the solution of the optimization-based policy~\eqref{eq:op_policy} and denote $\zeta^\star$ as the resulting trajectory. Assume $c(\cdot), c_H(\cdot), f(\cdot), g(\cdot)$ are twice differentiable in a neighborhood of $(\theta, \zeta^\star)$. Let
\begin{equation*}\label{eq:tau_matric}
\begin{aligned}
    & A = \nabla_{\zeta} \kappa(\zeta^\star; \theta),  \\
    & B = \nabla^2_{\theta \zeta} J(\zeta^\star;\theta) - \sum_{i=-1}^{H} \sum_{j=1}^{|\kappa_i|} \lambda_{i,j} \nabla^2_{\theta \zeta} [\kappa_i(\zeta^\star; \theta)]_j, \\
    & C = \nabla_{\theta} \kappa(\zeta^\star; \theta), \\
    & D = \nabla^2_{\zeta \zeta} J(\zeta^\star; \theta) - \sum_{i=-1}^{H} \sum_{j=1}^{|\kappa_i|} \lambda_{i,j} \nabla^2_{\zeta \zeta} [\kappa_i(\zeta^\star;\theta)]_j,
\end{aligned}
\end{equation*}
If $\text{rank}(A)=n_\kappa$ and $D$ is non-singular,
then the gradient $\nabla_{\theta} u_i^\star$ takes the following form, 
\begin{equation}\label{eq:u_grad}
    \nabla_{\theta} u^\star_{i} = [\nabla_{\theta} \zeta^\star]_{(n+m)\cdot i+n:(n+m)\cdot(i+1)},
\end{equation} 
with 
\begin{equation*}\label{eq:tau_grad}
\begin{aligned}
    & \nabla_{\theta} \zeta^\star = D^{-1}A^\tr (AD^{-1}A^\tr )^{-1} (AD^{-1}B-C)-D^{-1}B, 
\end{aligned}
\end{equation*}
where $n, m$ are the state and action dimensions, and the Lagrange multiplier $\lambda \in \R^{n_\kappa}$ satisfies $\lambda^\tr A = \nabla_{\zeta} J(\zeta^\star; \theta)$.
\end{proposition}
We note that $\lambda_{i,j}$ is the dual variable corresponding to the $j$-th element of $\kappa_i(\zeta^\star;\theta)$. 
Proposition \ref{prop:grad_policy} is a direct application of the constrained optimization differentiation~\cite{xu2024revisiting} and its proof is shown in Appendix~\cite{bian2025diffopreinforcementlearningoptimizationbased} for completeness.

We now turn to the estimation of the policy gradient.
Let $\tau = (x_0, u_0, x_1, u_1, \ldots, u_{T-1}, x_T)$ be the trajectory induced by the policy. The analytical policy gradient update rule would be the following:
\begin{proposition}[Policy gradient update]\label{prop:2}
    Consider the policy learning problem~\eqref{eq:problem}, the policy gradient takes the analytical form of: 
\begin{equation}\label{eq:pg}
\begin{aligned}
    \nabla_\theta C(\theta) 
    = \E \left[L(\tau) \left(\sum_{t=0}^{T} \frac{1}{\sigma^2} [\nabla_\theta u_{t}^\star]^\tr (u_{t}-u_{t}^\star)\right)\right],
\end{aligned}
\end{equation}
where $u_t, u^\star_{t}$ are the actual control action and corresponding optimal solution derived by the policy at time $t$.
\end{proposition}
\begin{proof}
    The result follows from the REINFORCE gradient estimator~\cite{sutton2018reinforcement}, applied to a stochastic policy modeled as a Gaussian centered at the optimization-based solution \( u_t^\star \). For a detailed derivation, please refer to Appendix~\cite{bian2025diffopreinforcementlearningoptimizationbased}.
\end{proof}
By establishing differentiability of \( u_t^\star \) with respect to \( \theta \) via the implicit function theorem as described in Proposition~\ref{prop:grad_policy}, we estimate the policy gradient in practice by Monte Carlo sampling over $N$ trajectories: 
\begin{equation}\label{eq:estimate_pg}
\scalebox{0.85}{$
\widehat{\nabla}_{\theta} C(\theta^{(k)}) = \frac{1}{N} \sum_{n=1}^{N} \left[ L(\tau^{(n)}) \sum_{t=0}^{T} \frac{1}{\sigma^2} \nabla_\theta u_t^{\star(n)\top} (u_t^{(n)} - u_t^{\star(n)}) \right]
$}
\end{equation}

It is worth noting that although we employ the REINFORCE algorithm~\cite{sutton2018reinforcement} in Algorithm~\ref{alg:diffop} to estimate the policy gradients, there are several techniques available to reduce the variance of the gradient estimator~\cite{zhao2011analysis, grathwohl2017backpropagation} which could be incorporated into the framework. 

\section{Non-Asymptotic Convergence Analysis}\label{sec:convergence}
We present our main theoretical result, which establishes the convergence properties of DiffOP when trained using Algorithm~\ref{alg:diffop}. Since the policy is implicitly defined through an optimization problem, the sensitivity of its solution with respect to the policy parameters is central to the analysis.

\subsection{Convergence Under Bounded Sensitivity}

We begin by stating a sufficient condition for convergence: bounded first- and second-order sensitivity of the optimization-defined policy with respect to its parameters.
\begin{assumption}[Bounded First/Second-Order Sensitivity]
\label{as:bound}
There exists a constant \( M_{u1} > 0, M_{u2} > 0 \) such that for all planning step $i$, the solution $u_i^\star(x; \theta)$ to the optimization-defined policy satisfies:
\[
\left\| \nabla_\theta u_i^\star(x; \theta) \right\| \leq M_{u1},  \,\,
\left\| \nabla^2_\theta u_i^\star(x; \theta) \right\| \leq M_{u2}.
\]
\end{assumption}

This assumption ensures that the policy responds smoothly to changes in the parameters $\theta$. Next, we assume standard conditions within the policy optimization~\cite{papini2018stochastic}, regarding the initial state distribution and the trajectory cost.
\begin{assumption}\label{as:3}
    The initial state distribution $\mathcal{D}$ is supported in a region with a finite radius $D_0$. 
    For any initial state $x_0 \sim \mathcal{D}$, the trajectory cost $L(\tau)$ is bounded, i.e., there exists a positive constant $M_L > 0$ such that $L(\tau) \leq M_L$.
\end{assumption}

Now we can characterize the convergence of the proposed algorithm. 

\begin{theorem}[Convergence of DiffOP with Policy Gradient]
\label{th:converge}
Suppose Assumptions~\ref{as:bound} and~\ref{as:3} hold. For any $\epsilon > 0$, and $\nu \in (0,1)$, define a smoothness constant
\[
L_C = M_LT(\sqrt{m}\beta M_{u2} + \frac{1}{\sigma^2}M_{u1}^2 + Tm\beta^2M_{u1}^2), 
\]
a stepsize $\eta = \frac{1}{4L_C}$, 
the number of policy iterations $K$, and the number of sampled trajectories for each policy gradient step $N= \frac{2m\beta^2M_L^2T^2M_{u1}^2}{\epsilon^2} \log \frac{2Kd}{\nu}$. Then, with probability at least $1-\nu$, we have
\begin{align}\label{eq:pg_convergence}
    \min_{k=0, \ldots, K-1} \|\nabla_\theta C (\theta^{(k)})\|^2 \leq \frac{16L_C(C(\theta^{(0)}) -C(\overline{\theta}))}{K} + 3\epsilon.
\end{align}
where $\overline{\theta}$ is the global optimum of \eqref{eq:problem}. 
\end{theorem}
\begin{proof}
The detailed proof is provided in Appendix~\cite{bian2025diffopreinforcementlearningoptimizationbased}, building on standard stochastic approximation techniques~\cite{papini2018stochastic} and adapted to the implicit policy structure induced by optimization-based control.
\end{proof}

\subsection{Sufficient Conditions for Bounded Sensitivity}\label{sec:strongconvex}

We now present a common class of optimization-based policies that satisfy the sensitivity bounds in Assumption~\ref{as:bound}. Specifically, we consider the unconstrained strongly convex policy~\eqref{eq:op_policy}, for which concrete sufficient conditions can be established.

Let \( x_0 = x_{\text{init}} \) denote the initial state, and let \( u = [u_0, \ldots, u_{H-1}] \in \R^{mH} \) be the control sequence generated by the optimization-based policy with parameters \( \theta \). The unrolled cost function can be written as:
\begin{equation}\label{eq:cost_policy}
     J(x_{\text{init}}, u, \theta)  = \sum_{i=0}^{H-1} c(x_i, u_i; \theta_c) + c_H(x_H; \theta_H),
\end{equation}
subject to the dynamics \( x_0 = x_{\text{init}}, \; x_{i+1} = f(x_i, u_i; \theta_f) \) for \( i = 0, \ldots, H-1 \).
We next introduce technical assumptions that characterize the optimization landscape and smoothness properties of this class of optimization-defined policies.
\begin{assumption}\label{as:1}
The function \( J(x_{\text{init}}, u, \theta) \) is \( \mu \)-strongly convex with respect to \( u \).
\end{assumption}
Strong convexity ensures the uniqueness of the optimal solution and the stability of the mapping from \( \theta \) to the control sequence \( u^\star(x; \theta) \). This injectivity yields well-defined gradients and the convergence guarantees in our policy optimization framework.

\begin{assumption}
\label{as:2}
Let \( z = (x_{\text{init}}, u, \theta) \) and \( \mathcal{Z} \subset \R^{n_z} \) be a compact set. The unrolled cost function \( J(z) \) satisfies the following:
\begin{itemize}
    \item \( \nabla_z J(z) \) is \( L_1 \)-Lipschitz: for any \( z, z' \in \mathcal{Z} \),
    \[
    \left\| \nabla_z J(z) - \nabla_z J(z') \right\| \leq L_1 \| z - z' \|.
    \]
    \item \( \nabla_\theta \nabla_u J(z) \) is \( L_2 \)-Lipschitz: for any \( z, z' \in \mathcal{Z} \),
    \[
    \left\| \nabla_\theta \nabla_u J(z) - \nabla_\theta \nabla_u J(z') \right\| \leq L_2 \| z - z' \|.
    \]
    \item \( \nabla_u^2 J(z) \) is \( L_3 \)-Lipschitz: for any \( z, z' \in \mathcal{Z} \),
    \[
    \left\| \nabla_u^2 J(z) - \nabla_u^2 J(z') \right\| \leq L_3 \| z - z' \|.
    \]
\end{itemize}
\end{assumption}

Assumption~\ref{as:2} captures the smoothness of the unrolled cost function and its partial derivatives with respect to the optimization variables and policy parameters. These conditions ensure the gradients and Hessians involved in the policy sensitivity analysis are well-behaved. Given Assumptions~\ref{as:1} and~\ref{as:2}, we can show that Assumption 2 is satisfied. 

\begin{proposition}
\label{prop:bounded_sensitivity}
Under Assumptions~\ref{as:1} and~\ref{as:2}, the optimization-defined policy \( \pi_\theta \) satisfies the bounded sensitivity condition stated in Assumption~\ref{as:bound}. In particular, there exist constants
\[
M_{u1} = \frac{L_1}{\mu}, \quad
M_{u2} = \frac{L_2}{\mu} + \frac{L_1 L_2 + L_1 L_3}{\mu^2} + \frac{L_1^2 L_3}{\mu^3},
\]
such that for all planning steps \( i \),
\[
\left\| \nabla_\theta u_i^\star(x; \theta) \right\| \leq M_{u1}, \quad
\left\| \nabla^2_\theta u_i^\star(x; \theta) \right\| \leq M_{u2}.
\]
\end{proposition}
\begin{proof}
    The result follows from applying the implicit function theorem to the first-order optimality conditions of the unconstrained problem. The complete proof is provided in Appendix~\cite{bian2025diffopreinforcementlearningoptimizationbased}. 
\end{proof}
Our analysis makes a key theoretical contribution by establishing the first non-asymptotic convergence guarantee for learning optimization-based policies via policy gradients. We show that DiffOP converges to an \(\epsilon\)-stationary point within \(\mathcal{O}(1/\epsilon)\) iterations under mild smoothness and strong convexity assumptions. Notably, this matches the standard convergence rate of first-order policy gradient methods, despite the added complexity of differentiating through an implicit optimization-based policy.

The central technical challenge lies in bounding the sensitivity of the optimization solution w.r.t. policy parameters. We address this by introducing structural assumptions on the inner optimization problem and leveraging analytical tools from bilevel optimization theory~\cite{ji2021bilevel, kwon2023fully}. This analysis not only underpins our convergence guarantee but also provides general insights into when optimization-based policies admit stable gradient-based learning. In more general settings, such as non-convex or constrained optimization problems, the solution mapping can become non-smooth or even discontinuous, especially near constraint boundaries or saddle points. This can lead to unbounded gradients and unstable updates during training, making theoretical convergence guarantees significantly harder to establish. For instance, in the presence of inequality constraints, small changes in parameters can trigger abrupt changes in the active set, breaking the smoothness assumptions required for standard implicit differentiation. Extending our convergence analysis to cover such cases remains an important and challenging direction for future work.


However, in practice, certain smooth non-convex constrained problems may exhibit locally bounded sensitivities~\cite{scieur2022curse} and satisfy Assumption~\ref{as:bound}. To demonstrate the practical performance of DiffOP, we include experiments on constrained control tasks that highlight its applicability to realistic settings.

\begin{table*}[t]
\centering
\begin{tabular}{lll}
\hline
Systems                     & Objective parameter $\theta_c, \theta_H$ & Dynamic parameter $\theta_f$               \\ \hline
Cartpole                    & \multirow{3}{*}{\begin{tabular}[c]{@{}l@{}}$c(x_i,u_i) = \|\theta_c^\tr (x_i-x_{\text{nom}})\|^2  + \|u_i\|^2$\\ $c_H(x_H) = \|\theta_H^\tr (x_H-x_{\text{nom}})\|^2$\end{tabular}} & length and mass for each link     \\
Two-link robot arm          &                                                                                                                                                                                     & length and mass for each link     \\
6-DoF quadrotor maneuvering &                                                                                                                                                                                     & mass, wing length, inertia matrix \\\hline
\end{tabular}
\caption{Nonlinear experimental settings~\cite{jin2020pontryagin}. }\label{tab:exp_nonlinear}
\end{table*}
\begin{figure*}[t]
    \centering
    \includegraphics[width=0.9\linewidth]{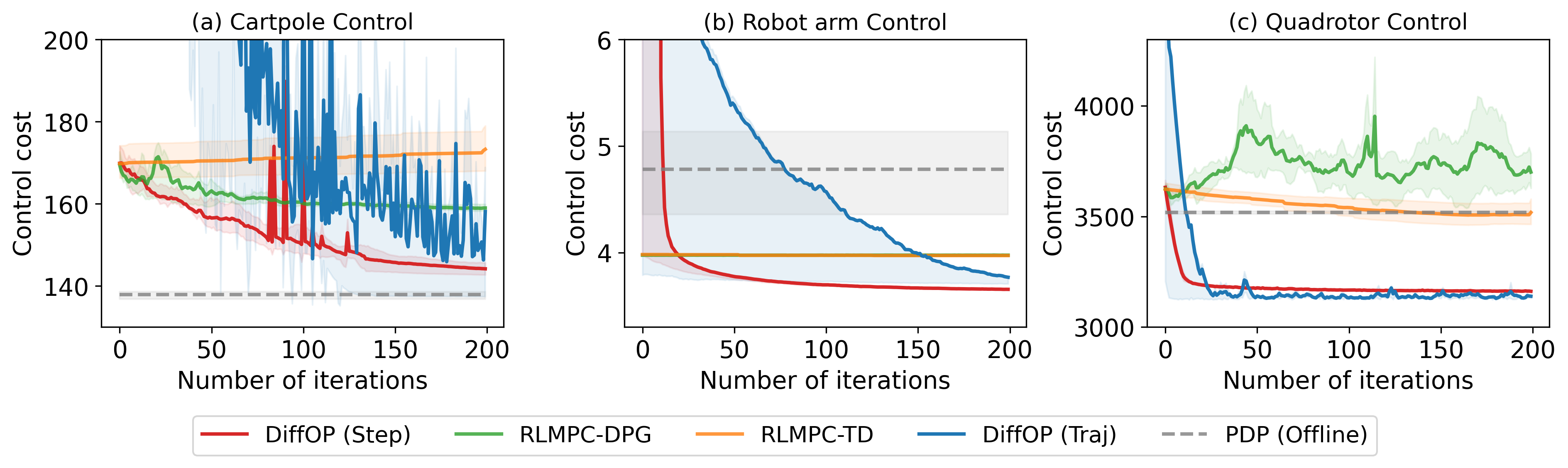}
    \caption{Control cost versus training iteration on nonlinear control tasks. Solid lines indicate the mean cost over 5 runs, and shaded regions denote the 20th to 80th percentile range.}
    \label{fig:nonlinear}
\end{figure*}
\section{Experiments}\label{sec:experiments}

In this section, we evaluate the effectiveness of DiffOP on nonlinear control tasks and a constrained voltage regulation problem. 
Detailed experimental settings are provided in Appendix~\cite{bian2025diffopreinforcementlearningoptimizationbased}. 

\subsection{Nonlinear Control Tasks}

In this section, we evaluate the proposed DiffOP framework on a set of nonlinear control tasks. We follow the experimental setup introduced in~\cite{jin2020pontryagin}, and summarize the environment details in Table~\ref{tab:exp_nonlinear}.

We compare against the following optimization-based control baselines. In all methods, the forward pass is solved using CasADi optimization solver~\cite{andersson2019casadi}:
\begin{itemize}
    \item \textbf{PDP (Offline)}~\cite{jin2020pontryagin}: Learns optimization parameters via Pontryagin Differentiable Programming from expert trajectories. As a purely offline method, it does not adapt online. Reported performance is taken directly from the original paper.
    
    \item \textbf{RLMPC-TD}~\cite{gros2019data}: A model-free reinforcement learning approach that updates a short-horizon MPC policy using temporal-difference (TD) Q-learning.

    \item \textbf{RLMPC-DPG}~\cite{seel2023combining}: A variant of RLMPC-TD that also employs deterministic policy gradient (DPG) for updating the MPC parameters, with the same short-horizon control structure.
\end{itemize}

We evaluate \textbf{DiffOP} under two deployment modes to highlight its flexibility:

\begin{itemize}
    \item \textbf{DiffOP (Step)}: At each time step, apply only the first action of the optimized trajectory, i.e., a receding-horizon scheme mirrors RLMPC-TD and RLMPC-DPG.

    \item \textbf{DiffOP (Traj)}: Generate the entire optimized control sequence once at the start of the episode and execute it open-loop. This trajectory-optimization setting boosts temporal consistency but sacrifices some robustness due to its open-loop nature.
\end{itemize}

For fair comparison, the planning horizon is set to 6 for Cartpole and 12 for RobotArm and Quadrotor across RLMPC-TD, RLMPC-DPG, and DiffOP (Step).  
For all DiffOP experiments, we set \(\beta = 0.01\), \(\eta = 0.1\) and $N = 10$.
RL-based MPC baselines are implemented using the MPC4RL framework~\cite{reinhardt2024mpc4rl}.
Results are presented in Figure~\ref{fig:nonlinear}. 
Compared to RL-based MPC baselines, DiffOP (Step) achieves faster convergence and lower final control cost across all nonlinear control tasks. In particular, it outperforms both RLMPC-TD and RLMPC-DPG which often converge to suboptimal solutions and may fail to consistently reduce control cost.

Compared to PDP (Offline), DiffOP (Traj) consistently achieves lower control cost in Robot Arm and Quadrotor tasks, showing the benefit of policy refinement through online interaction. While PDP relies solely on supervised expert data, DiffOP (Traj) fine-tunes the optimization-based policy using reinforcement learning, enabling it to surpass offline imitation performance and adapt better to the true system dynamics.
These results highlight the effectiveness of DiffOP under both receding-horizon and open-loop settings, and its robustness across control tasks with different levels of nonlinearity and complexity.

\subsection{Voltage Control with Power Injection Constraints}
We further evaluate DiffOP on a practical voltage control scenario involving reactive power constraints, using the IEEE 13-bus radial distribution benchmark.
We adopt the experimental setup of~\cite{feng2023bridging}, using the Pandapower toolbox~\cite{turitsyn2011options} to simulate the power network under two voltage disturbance scenarios: (1) high voltage due to excess PV generation under strong sunlight, and (2) low voltage caused by peak loads with insufficient generation. For each scenario, we vary load profiles to induce deviations of 5\%–15\% from the nominal voltage $v^{\text{nom}}$. The goal of voltage control is to achieve fast voltage restoration to an acceptable range (5\%) with minimal control.

Let \(v_i \in \mathbb{R}^3\) denote the squared voltage magnitudes, \(q_i \in \mathbb{R}^3\) the reactive power injection, and \(u_i \in \R^3\) the control inputs at step $i$. We define the optimization-based control policy for voltage regulation:
\begin{subequations}
\begin{align}
  \arg \min_u \quad & \sum_{i=0}^{H-1} \left( u_i^\top C_u u_i + {(v_i-v^{\text{nom}})}^\top C_v (v_i-v^{\text{nom}}) \right)  \notag \\
   & + {(v_H-v^{\text{nom}})}^\top C_v {(v_H-v^{\text{nom}})} \label{eq:voltage_obj} \\
  \text{s.t.} \quad & v_{i+1} = A q_i + v^{\text{env}},
  q_{i+1} = q_i + u_i, 
  \label{eq:voltage_dynamic}\\
  & \underline{q} \leq q_i \leq \overline{q} .\label{eq:voltage_power_con1}
\end{align}
\end{subequations}
The objective~\eqref{eq:voltage_obj} penalizes both voltage deviation from the nominal value \(v^{\text{nom}}\) and control effort, where \(C_u\) is a fixed parameter reflecting actuation cost.  
Equation~\eqref{eq:voltage_dynamic} models the voltage and reactive power dynamics, where $v^{\text{env}}$ is the initial voltage determined by load profiles. In general, $v^{\text{env}}$ is considered static due to time-scale separation between load change and the voltage control process.
Constraint~\eqref{eq:voltage_power_con1} ensures that the reactive power injection $q_t$ remains within operational limits.
The policy parameters include \(\theta = \{C_v, A\}\), which are learned to optimize both transient and steady-state performance—ensuring fast convergence, low operational cost, and stable voltage profiles.

\begin{table}[t]
\centering
\begin{tabular}{lcc}
\hline 
              & Transient cost & Steady state cost \\ \hline 
DiffOP (Step) &    \textbf{-6.81}       &    \textbf{-0.11}           \\ 
RLMPC-DPG       &      -6.11         &   -0.10       \\
RLMPC-TD        &      -4.62     &       -0.07      \\ \hline
DiffOP (Traj) &    \textbf{-6.80}            &      \textbf{-0.11}            \\
PDP (Offline) &     -5.86           &   -0.09                \\ \hline
Stable-DDPG   &    -5.61           &       -0.09            \\
TASRL         &    -6.76           &     \textbf{ -0.11}             \\ \hline
\end{tabular}
\caption{Performance on 500 scenarios of the IEEE 13-bus system. Lower cost indicates better performance. DiffOP (Traj) and PDP use horizon $H = 30$ with open-loop execution; DiffOP (Step), RLMPC-DPG, and RLMPC-TD use $H = 6$ with first-action execution.} \label{tab:voltage}
\end{table}

\begin{figure}[t]
    \centering
    \begin{minipage}[t]{0.22\textwidth}
        \centering
        \includegraphics[height=5.0cm]{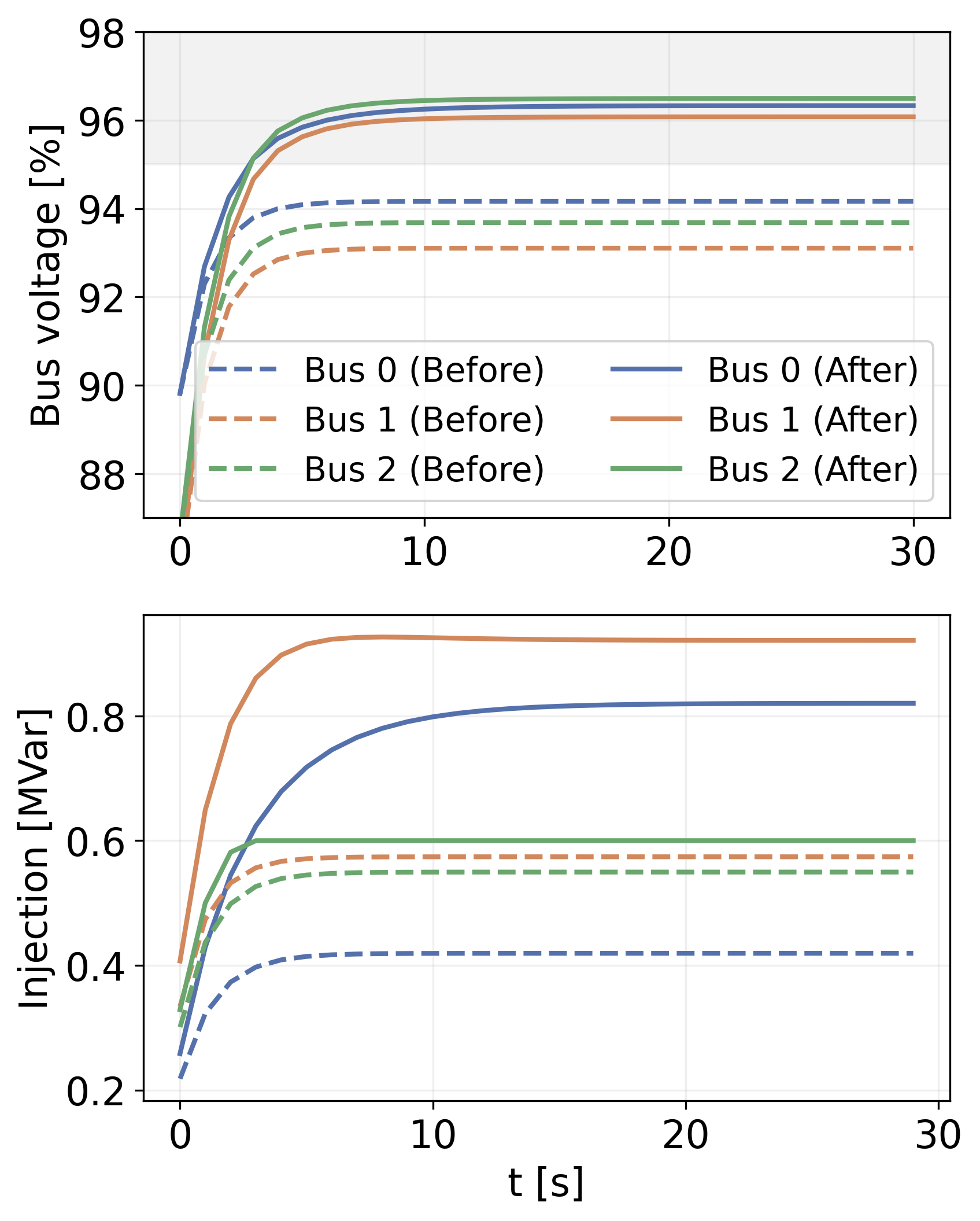}
        \caption*{(a) DiffOP(Traj)}
    \end{minipage}
    \hfill
    \begin{minipage}[t]{0.24\textwidth}
        \centering
        \includegraphics[height=5.0cm]{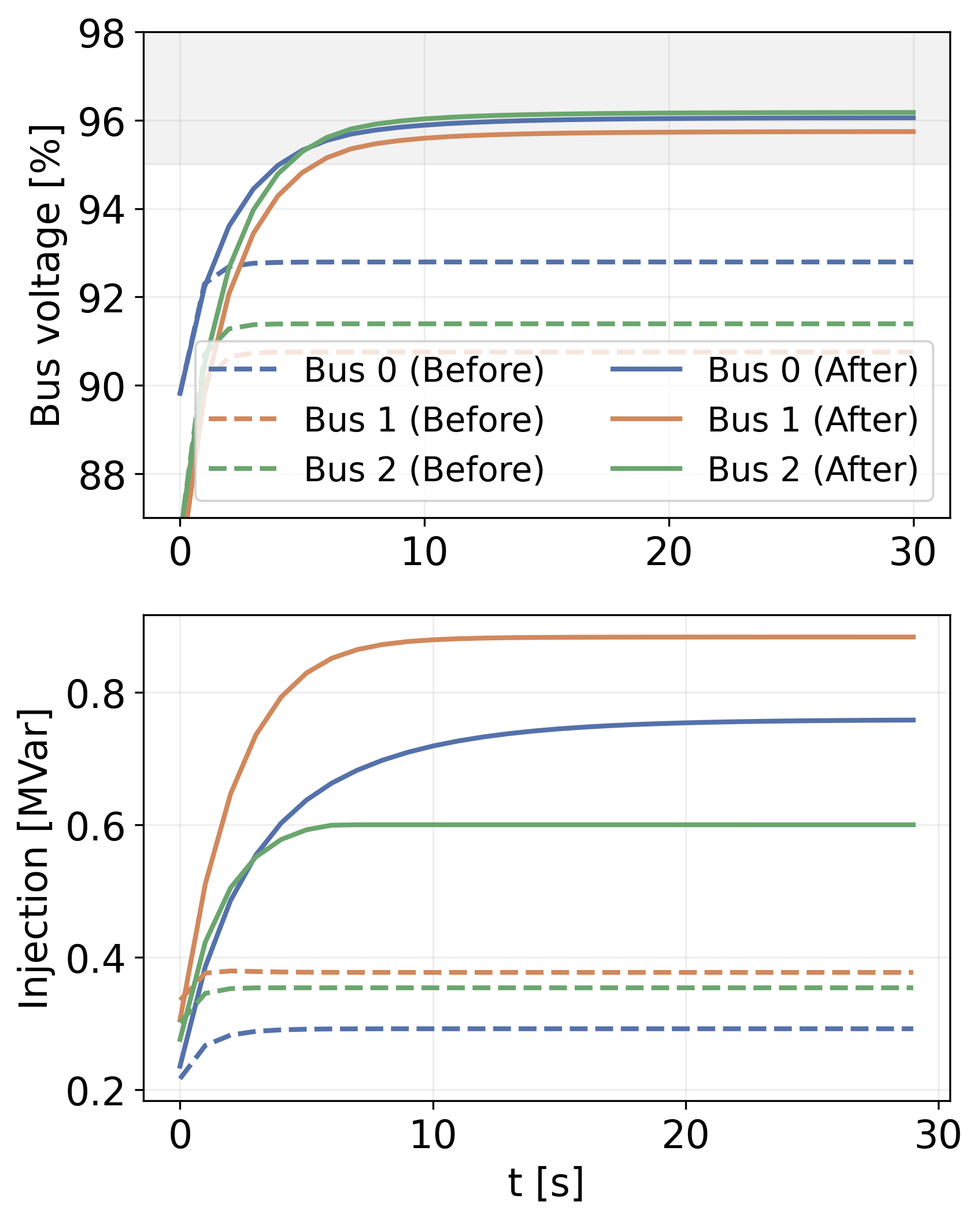}
        \caption*{(b) DiffOP(Step)}
    \end{minipage}
    \caption{Voltage and reactive power trajectories. Solid lines show post-training responses; dashed lines are pre-training performance. Gray shaded areas denote the safe operating bounds for system stability.}
    \label{fig:voltage}
\end{figure}

Our proposed DiffOP is evaluated in comparison with all prior baselines, as well as two state-of-the-art algorithms designed for voltage control:
\begin{itemize}
\item \textbf{Stable-DDPG}~\cite{shi2022stability}: An RL-trained monotone neural policy for voltage control with theoretical voltage restoration guarantees.
\item \textbf{TASRL}~\cite{feng2023bridging}: 
An RL-policy that incorporates Stable-DDPG and safe gradient flow to achieve fast voltage restoration while guaranteeing steady-state optimality with limited reactive-power capacity.
\end{itemize}
All optimization-based policies are trained on 10,000 trajectories, each of length $T = 30$ seconds. The planning horizon is set to $H=6$ across RLMPC-TD, RLMPC-DPG, and DiffOP (Step). The overall performance is summarized in Table~\ref{tab:voltage}.
For comparison, we use results reported in~\cite{feng2023bridging} for Stable-DDPG and TASRL.
DiffOP consistently achieves the lowest transient and steady-state costs across all baselines. 
Compared to MPC-based approaches, DiffOP effectively learns both cost and dynamics parameters, achieving the best performance in both transient and steady-state costs. Notably, it also outperforms RL-based neural network policies in transient cost and matches TASRL’s steady-state cost.


Further, we visualize the voltage and reactive power injection trajectories in Figure~\ref{fig:voltage}. 
After policy optimization, both variants exhibit significantly improved voltage regulation and faster convergence, with all voltage trajectories stabilized within the safe operating bounds (gray regions), in contrast to the untrained policy. 
Notably, our framework naturally supports applying a sequence of actions to the system, rather than being limited to receding-horizon execution. Despite the different planning settings used in DiffOP (Step) and DiffOP (Traj), both variants consistently learn effective control policies.

\section{Conclusion}\label{sec:conclusion}
In this paper, we proposed DiffOP, a novel framework for learning optimization-based control policies via analytical implicit policy gradients.
Experiments on nonlinear control and constrained voltage regulation tasks show that DiffOP consistently improves performance and outperforms both model-free and MPC-based baselines.
This work opens the door to several interesting directions for future research. First, we aim to explore improved sample complexity for learning the optimization-based policy. Secondly, it would be of significant interest to further investigate the robustness of DiffOP in comparison to model-free RL approaches. Finally, we envision extending the experimental results of DiffOP to a wide range of real-world deployments.

\bibliography{aaai25}

@misc{bian2025diffopreinforcementlearningoptimizationbased,
      title={DiffOP: Reinforcement Learning of Optimization-Based Control Policies via Implicit Policy Gradients}, 
      author={Yuexin Bian and Jie Feng and Yuanyuan Shi},
      year={2025},
      eprint={2411.07484},
      archivePrefix={arXiv},
      primaryClass={eess.SY},
      url={https://arxiv.org/abs/2411.07484}, 
}

@inproceedings{chen2019gnu,
  title={Gnu-rl: A precocial reinforcement learning solution for building hvac control using a differentiable mpc policy},
  author={Chen, Bingqing and Cai, Zicheng and Berg{\'e}s, Mario},
  booktitle={Proceedings of the 6th ACM international conference on systems for energy-efficient buildings, cities, and transportation},
  pages={316--325},
  year={2019}
}

@article{amos2018differentiable,
  title={Differentiable mpc for end-to-end planning and control},
  author={Amos, Brandon and Jimenez, Ivan and Sacks, Jacob and Boots, Byron and Kolter, J Zico},
  journal={Advances in neural information processing systems},
  volume={31},
  year={2018}
}

@inproceedings{chen2019optimal,
  title={Optimal Control Via Neural Networks: A Convex Approach},
  author={Chen, Yize and Shi, Yuanyuan and Zhang, Baosen},
  booktitle={International Conference on Learning Representations (ICLR)},
  year={2019}
}

@inproceedings{liu2018proximal,
  title={Proximal alternating direction network: A globally converged deep unrolling framework},
  author={Liu, Risheng and Fan, Xin and Cheng, Shichao and Wang, Xiangyu and Luo, Zhongxuan},
  booktitle={Proceedings of the AAAI Conference on Artificial Intelligence},
  volume={32},
  number={1},
  year={2018}
}

@inproceedings{ji2021bilevel,
  title={Bilevel optimization: Convergence analysis and enhanced design},
  author={Ji, Kaiyi and Yang, Junjie and Liang, Yingbin},
  booktitle={International Conference on Machine Learning (ICML)},
  pages={4882--4892},
  year={2021},
  organization={PMLR}
}

@inproceedings{kwon2023fully,
  title={A fully first-order method for stochastic bilevel optimization},
  author={Kwon, Jeongyeol and Kwon, Dohyun and Wright, Stephen and Nowak, Robert D},
  booktitle={International Conference on Machine Learning},
  pages={18083--18113},
  year={2023},
  organization={PMLR}
}

@article{xu2024revisiting,
  title={Revisiting implicit differentiation for learning problems in optimal control},
  author={Xu, Ming and Molloy, Timothy L and Gould, Stephen},
  journal={Advances in Neural Information Processing Systems},
  volume={36},
  year={2024}
}

@book{sutton2018reinforcement,
  title={Reinforcement learning: An introduction},
  author={Sutton, Richard S and Barto, Andrew G},
  year={2018},
  publisher={MIT press}
}

@inproceedings{zhang2021sample,
  title={Sample efficient reinforcement learning with REINFORCE},
  author={Zhang, Junzi and Kim, Jongho and O'Donoghue, Brendan and Boyd, Stephen},
  booktitle={Proceedings of the AAAI conference on artificial intelligence},
  volume={35},
  number={12},
  pages={10887--10895},
  year={2021}
}

@inproceedings{yang2021sample,
  title={Sample complexity of policy gradient finding second-order stationary points},
  author={Yang, Long and Zheng, Qian and Pan, Gang},
  booktitle={Proceedings of the AAAI Conference on Artificial Intelligence},
  volume={35},
  number={12},
  pages={10630--10638},
  year={2021}
}

@article{zhao2011analysis,
  title={Analysis and improvement of policy gradient estimation},
  author={Zhao, Tingting and Hachiya, Hirotaka and Niu, Gang and Sugiyama, Masashi},
  journal={Advances in Neural Information Processing Systems},
  volume={24},
  year={2011}
}

@article{grathwohl2017backpropagation,
  title={Backpropagation through the void: Optimizing control variates for black-box gradient estimation},
  author={Grathwohl, Will and Choi, Dami and Wu, Yuhuai and Roeder, Geoffrey and Duvenaud, David},
  journal={International Conference on Learning Representations (ICLR)},
  year={2017}
}

@inproceedings{papini2018stochastic,
  title={Stochastic variance-reduced policy gradient},
  author={Papini, Matteo and Binaghi, Damiano and Canonaco, Giuseppe and Pirotta, Matteo and Restelli, Marcello},
  booktitle={International Conference on Machine Learning (ICML)},
  pages={4026--4035},
  year={2018},
  organization={PMLR}
}

@article{diamond2016cvxpy,
  title={CVXPY: A Python-embedded modeling language for convex optimization},
  author={Diamond, Steven and Boyd, Stephen},
  journal={Journal of Machine Learning Research},
  volume={17},
  number={83},
  pages={1--5},
  year={2016}
}

@article{gill2005snopt,
  title={SNOPT: An SQP algorithm for large-scale constrained optimization},
  author={Gill, Philip E and Murray, Walter and Saunders, Michael A},
  journal={SIAM review},
  volume={47},
  number={1},
  pages={99--131},
  year={2005},
  publisher={SIAM}
}

@article{donti2017task,
  title={Task-based end-to-end model learning in stochastic optimization},
  author={Donti, Priya and Amos, Brandon and Kolter, J Zico},
  journal={Advances in neural information processing systems},
  volume={30},
  year={2017}
}

@inproceedings{srinivas2018universal,
  title={Universal planning networks: Learning generalizable representations for visuomotor control},
  author={Srinivas, Aravind and Jabri, Allan and Abbeel, Pieter and Levine, Sergey and Finn, Chelsea},
  booktitle={International Conference on Machine Learning (ICML)},
  pages={4732--4741},
  year={2018},
  organization={PMLR}
}

@article{jin2020pontryagin,
  title={Pontryagin differentiable programming: An end-to-end learning and control framework},
  author={Jin, Wanxin and Wang, Zhaoran and Yang, Zhuoran and Mou, Shaoshuai},
  journal={Advances in Neural Information Processing Systems},
  volume={33},
  pages={7979--7992},
  year={2020}
}

@book{bertsekas2014constrained,
  title={Constrained optimization and Lagrange multiplier methods},
  author={Bertsekas, Dimitri P},
  year={2014},
  publisher={Academic press}
}

@inproceedings{gros2021reinforcement,
  title={Reinforcement learning based on mpc and the stochastic policy gradient method},
  author={Gros, S{\'e}bastien and Zanon, Mario},
  booktitle={2021 American Control Conference },
  pages={1947--1952},
  year={2021},
  organization={IEEE}
}

@article{gros2019data,
  title={Data-driven economic NMPC using reinforcement learning},
  author={Gros, S{\'e}bastien and Zanon, Mario},
  journal={IEEE Transactions on Automatic Control},
  volume={65},
  number={2},
  pages={636--648},
  year={2019},
  publisher={IEEE}
}

@inproceedings{seel2023combining,
  title={Combining Q-learning and deterministic policy gradient for learning-based MPC},
  author={Seel, Katrine and Gros, S{\'e}bastien and Gravdahl, Jan Tommy},
  booktitle={2023 62nd IEEE Conference on Decision and Control},
  pages={610--617},
  year={2023},
  organization={IEEE}
}

@article{lawrence2025view,
  title={A view on learning robust goal-conditioned value functions: Interplay between RL and MPC},
  author={Lawrence, Nathan P and Loewen, Philip D and Forbes, Michael G and Gopaluni, R Bhushan and Mesbah, Ali},
  journal={arXiv preprint arXiv:2502.06996},
  year={2025}
}

@article{wachter2006implementation,
  title={On the implementation of an interior-point filter line-search algorithm for large-scale nonlinear programming},
  author={W{\"a}chter, Andreas and Biegler, Lorenz T},
  journal={Mathematical programming},
  volume={106},
  pages={25--57},
  year={2006},
  publisher={Springer}
}

@article{salzmann2023real,
  title={Real-time neural MPC: Deep learning model predictive control for quadrotors and agile robotic platforms},
  author={Salzmann, Tim and Kaufmann, Elia and Arrizabalaga, Jon and Pavone, Marco and Scaramuzza, Davide and Ryll, Markus},
  journal={IEEE Robotics and Automation Letters},
  volume={8},
  number={4},
  pages={2397--2404},
  year={2023},
  publisher={IEEE}
}

@book{machowski2020power,
  title={Power system dynamics: stability and control},
  author={Machowski, Jan and Lubosny, Zbigniew and Bialek, Janusz W and Bumby, James R},
  year={2020},
  publisher={John Wiley \& Sons}
}

@article{negenborn2008multi,
  title={Multi-agent model predictive control for transportation networks: Serial versus parallel schemes},
  author={Negenborn, Rudy R and De Schutter, Bart and Hellendoorn, Johannes},
  journal={Engineering Applications of Artificial Intelligence},
  volume={21},
  number={3},
  pages={353--366},
  year={2008},
  publisher={Elsevier}
}

@book{spong2020robot,
  title={Robot modeling and control},
  author={Spong, Mark W and Hutchinson, Seth and Vidyasagar, Mathukumalli},
  year={2020},
  publisher={John Wiley \& Sons}
}

@article{lambert2020objective,
  title={Objective mismatch in model-based reinforcement learning},
  author={Lambert, Nathan and Amos, Brandon and Yadan, Omry and Calandra, Roberto},
  journal={arXiv preprint arXiv:2002.04523},
  year={2020}
}

@inproceedings{jang2024active,
  title={Active reinforcement learning for robust building control},
  author={Jang, Doseok and Yan, Larry and Spangher, Lucas and Spanos, Costas J},
  booktitle={Proceedings of the AAAI Conference on Artificial Intelligence},
  volume={38},
  number={20},
  pages={22150--22158},
  year={2024}
}

@book{grune2017nonlinear,
  title={Nonlinear model predictive control},
  author={Gr{\"u}ne, Lars and Pannek, J{\"u}rgen and Gr{\"u}ne, Lars and Pannek, J{\"u}rgen},
  year={2017},
  publisher={Springer}
}

@article{morari1999model,
  title={Model predictive control: past, present and future},
  author={Morari, Manfred and Lee, Jay H},
  journal={Computers \& chemical engineering},
  volume={23},
  number={4-5},
  pages={667--682},
  year={1999},
  publisher={Elsevier}
}

@inproceedings{esfahani2021approximate,
  title={Approximate robust NMPC using reinforcement learning},
  author={Esfahani, Hossein Nejatbakhsh and Kordabad, Arash Bahari and Gros, S{\'e}bastien},
  booktitle={2021 European Control Conference},
  pages={132--137},
  year={2021},
  organization={IEEE}
}

@article{lin2023reinforcement,
  title={Reinforcement learning-based model predictive control for discrete-time systems},
  author={Lin, Min and Sun, Zhongqi and Xia, Yuanqing and Zhang, Jinhui},
  journal={IEEE Transactions on Neural Networks and Learning Systems},
  volume={35},
  number={3},
  pages={3312--3324},
  year={2023},
  publisher={IEEE}
}

@article{scieur2022curse,
  title={The curse of unrolling: Rate of differentiating through optimization},
  author={Scieur, Damien and Gidel, Gauthier and Bertrand, Quentin and Pedregosa, Fabian},
  journal={Advances in Neural Information Processing Systems},
  volume={35},
  pages={17133--17145},
  year={2022}
}

@article{turitsyn2011options,
  title={Options for control of reactive power by distributed photovoltaic generators},
  author={Turitsyn, Konstantin and Sulc, Petr and Backhaus, Scott and Chertkov, Michael},
  journal={Proceedings of the IEEE},
  volume={99},
  number={6},
  pages={1063--1073},
  year={2011},
  publisher={IEEE}
}

@inproceedings{shi2022stability,
  title={Stability constrained reinforcement learning for real-time voltage control},
  author={Shi, Yuanyuan and Qu, Guannan and Low, Steven and Anandkumar, Anima and Wierman, Adam},
  booktitle={2022 American Control Conference },
  pages={2715--2721},
  year={2022},
  organization={IEEE}
}

@article{feng2023bridging,
  title={Bridging transient and steady-state performance in voltage control: A reinforcement learning approach with safe gradient flow},
  author={Feng, Jie and Cui, Wenqi and Cort{\'e}s, Jorge and Shi, Yuanyuan},
  journal={IEEE Control Systems Letters},
  volume={7},
  pages={2845--2850},
  year={2023},
  publisher={IEEE}
}

@article{andersson2019casadi,
  title={CasADi: a software framework for nonlinear optimization and optimal control},
  author={Andersson, Joel AE and Gillis, Joris and Horn, Greg and Rawlings, James B and Diehl, Moritz},
  journal={Mathematical Programming Computation},
  volume={11},
  number={1},
  pages={1--36},
  year={2019},
  publisher={Springer}
}

@article{zhang2021trajectory,
  title={Trajectory tracking control of autonomous ground vehicles using adaptive learning MPC},
  author={Zhang, Kunwu and Sun, Qi and Shi, Yang},
  journal={IEEE Transactions on Neural Networks and Learning Systems},
  volume={32},
  number={12},
  pages={5554--5564},
  year={2021},
  publisher={IEEE}
}

@article{reiter2025synthesis,
  title={Synthesis of model predictive control and reinforcement learning: Survey and classification},
  author={Reiter, Rudolf and Hoffmann, Jasper and Reinhardt, Dirk and Messerer, Florian and Baumg{\~A}{\=I}rtner, Katrin and Sawant, Shamburaj and Boedecker, Joschka and Diehl, Moritz and Gros, Sebastien},
  journal={arXiv preprint arXiv:2502.02133},
  year={2025}
}

@inproceedings{reinhardt2024mpc4rl,
  title={MPC4RL-A Software Package for Reinforcement Learning based on Model Predictive Control},
  author={Reinhardt, Dirk and Baumg{\"a}rtner, Katrin and Frey, Jonathan and Diehl, Moritz and Gros, Sebastien},
  booktitle={2024 IEEE 63rd Conference on Decision and Control},
  pages={1787--1794},
  year={2024},
  organization={IEEE}
}

\appendix
\onecolumn

\setcounter{section}{0}
\renewcommand{\thesection}{\Alph{section}}
\setcounter{secnumdepth}{2} 

\section{Proof of Proposition~\ref{prop:grad_policy}: Gradient of the optimization-based policy}\label{ap:grad_op}
\begin{proof}
    Let $\lambda \in \R^{n_\kappa}$ and denote $\lambda_{i,j}$ as the dual variable corresponding to the $j$-th element of $\kappa_i$. 
    By the method of Lagrange multipliers~\cite{bertsekas2014constrained}, we form the Lagrangian:
    \begin{align*}
        L(\theta, \zeta, \lambda) = J(\zeta; \theta) - \sum_{i=-1}^{H}\sum_{j=1}^{|\kappa_i|} \lambda_{i,j} [\kappa_i(\zeta; \theta)]_j.
    \end{align*}
    Since the $\zeta^\star$ is the optimal solution, we have
    \begin{equation}\label{eq:kkt}
        \begin{bmatrix}
            \nabla_\zeta J(\zeta^\star; \theta) - \sum_{i=-1}^{H}\sum_{j=1}^{|\kappa_i|} \lambda_{i,j} \nabla_\zeta [\kappa_i(\zeta^\star; \theta)]_j \\
            \kappa(\zeta^\star; \theta)
        \end{bmatrix}
         = 0.
    \end{equation}
    For the first row in equation~\eqref{eq:kkt}, we have
    \begin{equation}
        \nabla_\zeta J(\zeta^\star; \theta) = \sum_{i=-1}^{H}\sum_{j=1}^{|\kappa_i|}
        \lambda_{i,j} \nabla_\zeta [\kappa_i(\zeta^\star; \theta)]_j = \lambda^\tr A,
    \end{equation}
    for $A$ defined as $A = \nabla_\zeta \kappa(\zeta^\star; \theta)$. 
    In the following statement, we simplify $[\kappa_i(\zeta^\star;\theta)]_j$ as $[\kappa_i]_j$ and $J(\zeta^\star;\theta)$ as $J$.
Then, differentiating the gradient of the Lagrangian with
respect to $\theta$ we have
\begin{equation}
    \begin{bmatrix}
        \nabla^2_{\theta \zeta} J + \nabla^2_\zeta J \nabla_\theta \zeta^\star - \nabla_\zeta \kappa^\tr \nabla_\theta \lambda - \sum_{i=-1}^{H}\sum_{j=1}^{|\kappa_i|}
        \lambda_{i,j} \left(\nabla_{\theta \zeta}^2 [\kappa_i]_j + \nabla_{\zeta}^2 [\kappa_i]_j \nabla_\theta \zeta^\star \right)  \\
        \nabla_\theta \kappa + \nabla_\zeta \kappa \nabla_\theta \zeta^\star
    \end{bmatrix}
    = 0.
\end{equation}
Therefore, we have
\begin{equation}
    \begin{bmatrix}
       \nabla_\zeta^2 J - \sum_{i=-1}^{H}\sum_{j=1}^{|\kappa_i|}
       \lambda_{ij} \nabla_\zeta^2 [\kappa_i]_j & -\nabla_\zeta \kappa^\tr  \\
        \nabla_\zeta \kappa & 0
    \end{bmatrix}
    \begin{bmatrix}
        \nabla_\theta \zeta^\star \\
        \nabla_\theta \lambda 
    \end{bmatrix}
    = 
    -
    \begin{bmatrix}
        \nabla_{\theta \zeta}^2 J - \sum_{i=-1}^{H}\sum_{j=1}^{|\kappa_i|}
        \lambda_{i,j} \nabla_{\theta \zeta}^2 [\kappa_i]_j \\
        \nabla_\theta \kappa
    \end{bmatrix}
\end{equation}
where all functions are evaluated at $(\zeta^\star, \theta)$. Then we can solve $\nabla_\theta \zeta^\star$ with
\begin{align*}
    \nabla_\theta \zeta^\star = D^{-1}A^\tr (AD^{-1}A^\tr )^{-1} (AD^{-1}B-C)-D^{-1}B.
\end{align*}
Since $\nabla_\theta \zeta^\star = [\nabla_\theta x^\star_0, \nabla_\theta u^\star_0, \ldots, \nabla_\theta x^\star_{H-1},\nabla_\theta u^\star_{H-1}, \nabla_\theta X_H^\star]$. 
After evaluating $\nabla_\theta \zeta^\star$, we have
\begin{align*}
    \nabla_{\theta} u^\star_{i} = [\nabla_{\theta} \zeta^\star]_{(n+m)\cdot i+n:(n+m)\cdot(i+1)}.
\end{align*}
\end{proof}

\section{Proof of Proposition~\ref{prop:2}: Analytical form of policy gradients}\label{ap:1}
Here we proof for proposition~\ref{prop:2}.
\begin{proof}
    With~\eqref{eq:sample}:
    \begin{align*}
     u_t \sim \pi_\theta (u|x_t) =  \frac{\phi(u|u_t^\star, \sigma^2 I)}{ \int_{u_t^\star-\beta \sigma^2I}^{u_t^\star+\beta \sigma^2I} \phi(u|u_t^\star, \sigma^2 I) \text{d}u},
    \end{align*}
    we have
    \begin{align*}
        \pi_\theta (u|x_t) = \frac{1}{Z (2\pi)^{\frac{m}{2}}|\sigma^2 I|^{\frac{1}{2}}} \exp{\left(-\frac{1}{2\sigma^2}(u_t-u_t^\star)^\tr (u_t-u_t^\star)\right)},
    \end{align*}
    where $Z$ is a normalization constant, i.e., the integral of the multivariate Gaussian PDF over the truncated range
    \begin{align*}
        Z = \int_{u_t^\star-\beta \sigma^2I}^{u_t^\star+\beta \sigma^2I} \phi(u|u_t^\star, \sigma^2 I) \text{d}u.
    \end{align*}
    The derivative of the log probability is
    \begin{align*}
        \nabla_\theta \log \pi_\theta(u|x_t) = \frac{1}{\sigma^2} [\nabla_\theta u_t^\star]^\tr (u_t - u_t^\star),
    \end{align*}
    where $u_t$ is the action applied to the system at time $t$, $u_t^\star$ is the corresponding optimal solution. In conjunction with $\nabla_\theta C(\theta) = \E [L(\tau) \nabla_\theta \log \pi_\theta(\tau)]$, we obtain
    \begin{align*}
    \nabla_\theta C(\theta) 
    = \E \left[L(\tau) \left(\sum_{t=0}^{T} \frac{1}{\sigma^2}\nabla_\theta {u_{t}^\star}^\tr 
    (u_{t}-u_{t}^\star)\right)\right].
    \end{align*}
\end{proof}

\section{Main Convergence Results of DiffOP}\label{ap:3}

\subsection{Proof of Theorem~\ref{th:converge}}\label{ap:proof_th1}
To prove Theorem~\ref{th:converge}, we first show that bounded policy derivatives imply the \(L\)-smoothness of the objective \(C(\theta)\), and then establish the convergence of the proposed policy gradient algorithm.

\begin{lemma}[Smoothness of $C(\theta)$]\label{lm:smooth_c}
    Suppose Assumptions~\ref{as:1} and \ref{as:2} hold. Then, we have, for any $\theta, \theta' \in \mathcal{G}_\theta$,
    $$\|\nabla_\theta C(\theta) -\nabla_\theta C(\theta')\| \leq L_C\|\theta-\theta'\|,$$
    where the constant $L_C$ is given by
    \begin{align*}
        L_C = M_LT(\sqrt{m}\beta M_{u2} + \frac{1}{\sigma^2}M_{u1}^2 + Tm\beta^2M_{u1}^2).
    \end{align*}
\end{lemma}
\begin{proof}
    We begin by establishing the boundedness of the gradient $\|\nabla_\theta \log \pi_\theta(u_t|x_t)\|$ and the Hessian $\|\nabla_\theta^2 \log \pi_\theta(u_t|x_t)\|$ w.r.t. $\theta$.
   
    Recall that
    \begin{align*}
        \pi_\theta (u_t|x_t) = \frac{1}{Z(2\pi)^{\frac{m}{2}}|\sigma^2 I|^{\frac{1}{2}}} \exp{\left(-\frac{1}{2\sigma^2}(u_t-u_t^\star(\theta))^\tr (u_t-u_t^\star(\theta))\right)},
    \end{align*}
    the gradient and the hessian are
    \begin{equation}
    \begin{aligned}
       & \nabla_\theta \log \pi_\theta(u_t|x_t) = \frac{1}{\sigma^2}\nabla_\theta u_t^\star(\theta)^\tr (u_t-u_t^\star(\theta)), \\
       & \nabla_\theta^2 \log \pi(u_t|x_t) = \frac{1}{\sigma^2}(\nabla_\theta^2 u_t^\star(\theta)^\tr (u_t-u_t^\star(\theta)) - \nabla_\theta u_t^\star(\theta)^\tr \nabla_\theta u_t^\star(\theta)).
    \end{aligned}
    \end{equation}
    Recall that as we use the truncated Gaussian policy, we have
    \begin{equation}
    \begin{aligned}
        &  -\beta \sigma^2 \leq [u_t^\star]_i - [u_t]_i  \leq \beta \sigma^2,  i = 1, \ldots, m \\
    \Leftrightarrow \quad   &         \|u_t - u_t^\star(\theta)\| \leq \sqrt{m} \beta \sigma^2 .
    \end{aligned}
    \end{equation}
    In conjunction with Assumption~\ref{as:bound}, we obtain
    \begin{equation}\label{eq:1glog}
        \|\nabla_\theta \log \pi_\theta(u_t|x_t) \| \leq 
        \frac{\sqrt{m}\beta \sigma^2}{\sigma^2} \|\nabla_\theta u_t^\star(\theta) \|  \leq
        \sqrt{m}\beta  M_{u1},
    \end{equation}
    and
    \begin{equation}\label{eq:2glog}
    \begin{aligned}
                \|\nabla_\theta^2 \log \pi_\theta (u_t|x_t) \| 
    &   \leq \frac{1}{\sigma^2} (\sqrt{m}\beta \sigma^2 \|\nabla_\theta^2 u_t^\star(\theta)\| + \|\nabla_\theta u_t^\star(\theta)\|^2) 
 \leq    
        \sqrt{m}\beta M_{u2} + \frac{1}{\sigma^2} M_{u1}^2.
    \end{aligned}
    \end{equation}

    Recall that $C(\theta) = \E_\tau [L(\tau)]$, we have 
    \begin{equation}\label{eq:1gC}
        \nabla_\theta C(\theta) = \int_\tau L(\tau) \pi_\theta(\tau) \nabla_\theta \log \pi_\theta(\tau) \text{d}\tau.
    \end{equation}
    By taking the derivative of~\eqref{eq:1gC}, we obtain
    \begin{equation}\label{eq:23}
        \nabla_\theta^2 C(\theta) = \int_\tau \left(L(\tau)\pi_\theta (\tau)\nabla^2_\theta \log \pi_\theta(\tau) + L(\tau)\pi_\theta(\tau) \nabla_\theta \log \pi_\theta(\tau) \nabla_\theta \log \pi_\theta(\tau)^\tr\right) \text{d}\tau 
    \end{equation}
    In this case, 
    \begin{align*}
        \nabla_\theta \log \pi_\theta(\tau) =  \sum_{t=1}^{T} 
        \nabla_\theta \log \pi_\theta(u_t|x_{t}), \nabla_\theta^2 \log \pi_\theta(\tau) =  \sum_{t=1}^{T} 
       \nabla_\theta^2 \log \pi_\theta(u_t|x_t).
    \end{align*}
    With equation~\eqref{eq:23} and Assumption~\ref{as:3}, we have
    \begin{equation}
    \begin{aligned}
        \|\nabla^2_\theta C(\theta)\|
        & \leq M_L 
        \max (\|\nabla_\theta^2 \log \pi_\theta(\tau)\| + \|\nabla_\theta \log \pi_\theta(\tau)\|^2) \underbrace{\int_\tau \pi_\theta(\tau) \text{d}\tau}_{=1}  \leq M_L T\left( \sqrt{m}\beta M_{u2} + \frac{1}{\sigma^2} M_{u1}^2 + Tm\beta^2 M_{u1}^2  \right).
    \end{aligned}
    \end{equation}
    Since the Hessian is bounded, for any $\theta, \theta' \in \mathcal{G}_\theta$, $C(\theta)$ satisfy
    \begin{equation}
       \| \nabla_\theta C(\theta) - \nabla_\theta C(\theta')\| \leq L_C \|\theta - \theta'\|,
    \end{equation}
    where
    \begin{align*}
        L_C = M_L T\left( \sqrt{m}\beta M_{u2} + \frac{1}{\sigma^2} M_{u1}^2 + Tm\beta^2 M_{u1}^2   \right).
    \end{align*}
    
\end{proof}

We then characterize the gradient estimation error $\|\widehat{\nabla}_\theta C(\theta^{(k)}) - \nabla_\theta C(\theta^{(k)})\|$, where $\theta^{(k)}$ is the policy parameter at $k$-th iteration and $\widehat{\nabla}_\theta C$ is estimated by the REINFORCE estimator (score function gradient estimator).

\begin{lemma}\label{lm:pg}
    Suppose Assumptions~\ref{as:1}, \ref{as:2} and \ref{as:3} hold. Then given $e_{grad}$, for any $\nu \in (0,1)$, when $N \geq \frac{2m\beta^2M_L^2T^2M_{u1}^2}{e_{grad}^2} \log \frac{2d}{\nu}$, then with probability at least $1-\nu$,
    \begin{equation}
        \|\widehat{\nabla}_\theta C(\theta^{(k)}) - \nabla_\theta C(\theta^{(k)})\|^2 \leq e_{grad}.
    \end{equation}
\end{lemma}

\begin{proof}
    Recall that
    \begin{align*}
        \nabla_\theta C(\theta) = \E[L(\tau)\nabla_\theta \log \pi_\theta(\tau)], \,\,
        \widehat{\nabla}_\theta C(\theta) = \frac{1}{N}\sum_{i=1}^N L(\tau^{(i)}) \nabla_\theta \log \pi_\theta (\tau^{(i)}).
    \end{align*}
    Let $X_i =L(\tau^{(i)})  \nabla_\theta \log \pi_\theta (\tau^{(i)})$.
    We have
    \begin{equation}
        \|X_i\|  \leq \|L(\tau^{(i)}) \nabla_\theta \log \pi_\theta (\tau^{(i)}) \| 
         \leq \|L(\tau^{(i)}\|\|\nabla_\theta \log \pi_\theta (\tau^{(i)})\| 
        \leq M_L \cdot \sqrt{m}\beta T M_{u1},
    \end{equation}
    where the third inequality we used equation~\eqref{eq:1glog}.
    
    Choose $N \geq \frac{2m\beta^2M_L^2T^2M_{u1}^2}{e_{grad}^2} \log \frac{2d}{\nu}$, 
    and by Hoeffding's bound, with probability at least $1-\nu$,
    \begin{equation}
        \|\widehat{\nabla}_\theta C(\theta^{(k)}) - \nabla_\theta C(\theta^{(k)})\|^2 \leq e_{grad}.
    \end{equation}
\end{proof}


We now can provide proof for Theorem~\ref{th:converge}.
\begin{proof}[Proof for Theorem~\ref{th:converge}]
    Let $\F_k$ be the filtration generated by $\{\widehat{\nabla}_\theta C(\theta^{(k')})\}_{k'=0}^{k-1}.$ Then we have $\theta^{(k)}$ is $\F_k$ measurable. We define the following event,
    \begin{equation}
        \mathcal{E}_k = \{\|\widehat{\nabla}_\theta C(\theta^{(k')}) - \nabla_\theta  C(\theta^{(k')})\|^2 \leq e_{grad}, \forall k' = 0, 1, \ldots, k-1\},
    \end{equation}
    and $\mathcal{E}_k$ is also $\F_k$-measurable. By lemma~\ref{lm:pg} and selection of $N \geq \frac{2m\beta^2M_L^2T^2M_{u1}^2}{e_{grad}^2} \log \frac{2Kd}{\nu}$, 
    we have $\|\widehat{\nabla}_\theta C(\theta^{(k)}) - \nabla_\theta C(\theta^{(k)})\|^2 \leq e_{grad}$ with probability at least $1-\frac{\nu}{K}$. Then we have
    \begin{equation}
        \E[\textbf{1}(\mathcal{E}_{k+1}|\F_k)\mathbf{1}(\mathcal{E}_k)] \geq (1-\frac{\nu}{K})\mathbf{1}(\mathcal{E}_k).
    \end{equation}
    Take expectation on both side, we obtain
    \begin{equation}
        \P(\mathcal{E}_{k+1}) = \P(\mathcal{E}_{k+1} \cap \mathcal{E}_k) = \E[\E[\textbf{1}(\mathcal{E}_{k+1}|\F_k)\mathbf{1}(\mathcal{E}_k)]] \geq (1-\frac{\nu}{K})\P(\mathcal{E}_k).
    \end{equation}
    As a result, we have, $\P(\mathcal{E}_K) \geq (1-\frac{\nu}{K})^K \P(\mathcal{E}_0) > 1-\nu$. Recall that
    \begin{align*}
        \theta^{(k+1)} =  \theta^{(k)} - \eta \widehat{\nabla}_\theta C(\theta^{(k)}).
    \end{align*}
    On the event $\mathcal{E}_K$ and based on the smoothness of the function $C(\theta)$ established in Lemma~\ref{lm:smooth_c}, we have
    \begin{align*}
        C(\theta^{(k+1)}) & \leq C(\theta^{(k)}) + \underbrace{\left<\nabla_\theta C(\theta^{(k)}), \theta^{(k+1)}-\theta^{(k)}\right>}_{e_1} + \underbrace{\frac{L_C}{2}\|\theta^{(k+1)}-\theta^{(k)}\|^2}_{e_2}
    \end{align*}
    For $e_1$, we have
    \begin{align*}
        e_1 & = \left< \nabla_\theta C(\theta^{(k)}), - \eta \widehat{\nabla}_\theta C(\theta^{(k)} \right> \\
        & \leq -\eta \left< \nabla_\theta C(\theta^{(k)}),  \widehat{\nabla}_\theta C(\theta^{(k)} \right> + \frac{\eta}{2} \|\nabla_\theta C(\theta^{(k)})\|^2 \\
        & = -\eta \left< \nabla_\theta C(\theta^{(k)}),  \widehat{\nabla}_\theta C(\theta^{(k)} \right> + \frac{\eta}{2} \|\nabla_\theta C(\theta^{(k)})\|^2 + \frac{\eta}{2} \|\widehat{\nabla}_\theta C(\theta^{(k)})\|^2 - \frac{\eta}{2} \|\widehat{\nabla}_\theta C(\theta^{(k)})\|^2 \\
        & = \frac{\eta}{2} \|\widehat{\nabla}_\theta C(\theta^{(k)}) - {\nabla}_\theta C(\theta^{(k)})\|^2 - \frac{\eta}{2} \|\widehat{\nabla}_\theta C(\theta^{(k)})\|^2 
    \end{align*}
    For $e_2$, we have
    \begin{align*}
        e_2 & = \frac{\eta^2 L_C}{2} \|\widehat{\nabla}_\theta C(\theta^{(k)})\|^2 \\
        & \leq \frac{\eta^2 L_C}{2} \left(\|\widehat{\nabla}_\theta C(\theta^{(k)})\|^2 + \|\widehat{\nabla}_\theta C(\theta^{(k)}) - 2 \nabla_\theta C(\theta^{(k)})\|^2\right) \\
        & = \frac{\eta^2 L_C}{2} \left(2\|\widehat{\nabla}_\theta C(\theta^{(k)})\|^2+ 4 \|\nabla_\theta C(\theta^{(k)})\|^2 -4\left<\widehat{\nabla}_\theta C(\theta^{(k)}), {\nabla}_\theta C(\theta^{(k)})\right>\right) \\
        & = \frac{\eta^2 L_C}{2} \left(2\|\nabla_\theta C(\theta^{(k)})\|^2 + 2 \|\widehat{\nabla}_\theta C(\theta^{(k)}) - {\nabla}_\theta C(\theta^{(k)})\|^2\right).
    \end{align*}
    Then we have
    \begin{equation}
    \begin{aligned}
        C(\theta^{(k+1)}) & \leq C(\theta^{(k)}) + e_1 + e_2 \\
        & \leq C(\theta^{(k)}) - (\frac{\eta}{2}-\eta^2L_C)\|\nabla_\theta C(\theta^{(k)})\|^2 +(\frac{\eta}{2}+\eta^2L_C)\|\widehat{\nabla}_\theta C(\theta^{(k)}) - \nabla_\theta C(\theta^{(k)})\|^2,
    \end{aligned}
    \end{equation}
    which, combined with Lemma~\ref{lm:pg}, yields that with probability at least $1-\nu$, 
    \begin{equation}\label{eq:c_smooth2}
        C(\theta^{(k+1)}) \leq C(\theta^{(k)}) - (\frac{\eta}{2}-\eta^2L_C)\|\nabla_\theta C(\theta^{(k)})\|^2 +(\frac{\eta}{2}+\eta^2L_C)e_{grad}.
    \end{equation}
    Telescoping equation~\eqref{eq:c_smooth2} over $k$ from $0$ to $K-1$ yields
    \begin{equation}\label{eq:31}
       \frac{1}{K} (\frac{1}{2}-\eta L_C) \sum_{k=0}^{K-1} \|\nabla_\theta C(\theta^{(k)})\|^2 \leq \frac{C(\theta^{(0)}) - C(\theta^\star)}{\eta K} + (\frac{1}{2} + \eta L_C) e_{grad}.
    \end{equation}
    Substituting $\eta = \frac{1}{4L_C}, \epsilon=e_{grad}$ in equation~\eqref{eq:31} yields
    \begin{equation}
        \frac{1}{K} \sum_{k=0}^{K-1} \|\nabla C_\theta (\theta^{(k)})\|^2 \leq \frac{16L_C(C(\theta^{(0)}) -C(\theta^\star))}{K} + 3\epsilon.
    \end{equation}
    This leads to
    \begin{equation}
        \min_k \|\nabla C_\theta (\theta^{(k)})\|^2 \leq \frac{16L_C(C(\theta^{(0)}) -C(\theta^\star))}{K} + 3\epsilon.
    \end{equation}
\end{proof}

\subsection{Proof of Proposition~\ref{prop:bounded_sensitivity} }\label{ap:proof_prop3}
In what follows, we make the parameter dependence explicit by writing \( u^\star(\theta) \) instead of \( u^\star \), to facilitate gradient analysis. We then characterize the Lipschitz properties of \( u^\star(\theta) \) and its gradient \( \nabla_\theta u^\star(\theta) \), which establishes Proposition~\ref{prop:bounded_sensitivity}.


\begin{proof}
The optimization policy, i.e., implicit function is defined as
\begin{align*}
    u^\star(\theta) = \min_u J(x_{\text{init}}, u, \theta).
\end{align*}
Due to the optimality condition, we have $\nabla_u J(x_{\text{init}}, u^\star(\theta), \theta) = 0.$ By taking derivative on both sides, using the chain rule and the implicit function theorem~\cite{rudin1964principles}, we obtain
\begin{align*}
    \nabla_\theta \nabla_u J(x_{\text{init}}, u^\star(\theta), \theta) + \nabla^2_{u} J(x_{\text{init}}, u^\star(\theta), \theta) \nabla_\theta u^\star(\theta) = 0.
\end{align*}
Let $z^\star = (x_{\text{init}}, u^\star(\theta), \theta)$, we have
\begin{equation}\label{eq:utheta}
    \nabla_\theta u^\star(\theta) = -[\nabla_u^2 J(z^\star)]^{-1} [\nabla_\theta \nabla_u J(z^\star)].
\end{equation}
    As $J(x_{\text{init}}, u, \theta)$ is $\mu$-strongly convex w.r.t. $u$, we have 
    \begin{align*}
        \|\nabla_u^2 J(z^\star)\|^{-1} \leq 1/\mu.
    \end{align*}
    Additionally, as $\nabla_z J(z^\star)$ is $L_1$-Lipschitz,
    we have
    \begin{align*}
      \|  \nabla_u J(z^\star) - \nabla_u J({z^\star}') \| \leq \|\nabla_z J(z^\star) - \nabla_z J({z^\star}') \| \leq L_1 \|z^\star-{z^\star}'\|.
    \end{align*}
    Then $\nabla_u J(z^\star)$ is $L_1$-Lipschitz, its partial derivative is bounded by $L_1$:
     \begin{align*}
         \|\nabla_\theta \nabla_u J(z^\star)\| \leq L_1.
     \end{align*}
     We obtain
    \begin{align*}
        \|\nabla_\theta u^\star (\theta)\| \leq  \|\nabla_u^2 J(z^\star)\|^{-1}\|\nabla_\theta \nabla_u J(z^\star)\| \leq \frac{L_1}{\mu}, \,\,
        \|u^\star(\theta) - u^\star(\theta') \| \leq \frac{L_1}{\mu} \|\theta - \theta'\|.
    \end{align*}
    
    For any $\theta, \theta' \in \R^d$, let ${z^\star}' = (x_{\text{init}}, u^\star(\theta'), \theta')$,
    we have
    \begin{equation}\label{eq:l1}
    \begin{aligned}
        \|\nabla_{\theta} u^\star(\theta) - \nabla_{\theta} u^\star(\theta')\| 
        & = \|-[\nabla_u^2 J(z^\star)]^{-1}[\nabla_\theta \nabla_u J(z^\star)] + [\nabla_u^2 J(z^\star)]^{-1}[\nabla_\theta \nabla_u J({z^\star}')] \\
        & \quad - [\nabla_u^2 J(z^\star)]^{-1}[\nabla_\theta \nabla_u J({z^\star}')] + [\nabla_u^2 J({z^\star}')]^{-1}[\nabla_\theta \nabla_u J({z^\star}')] \| \\
        & \leq \|-[\nabla_u^2 J(z^\star)]^{-1}[\nabla_\theta \nabla_u J(z^\star)] + [\nabla_u^2 J(z^\star)]^{-1}[\nabla_\theta \nabla_u J({z^\star}')] \| \\
        & \quad + \|- [\nabla_u^2 J(z^\star)]^{-1}[\nabla_\theta \nabla_u J({z^\star}')] + [\nabla_u^2 J({z^\star}')]^{-1}[\nabla_\theta \nabla_u J({z^\star}')] \| \\
        & \leq \|[\nabla_u^2 J(z^\star)]^{-1}\| \|\nabla_{\theta} \nabla_u J(z^\star) - \nabla_{\theta} \nabla_u J({z^\star}')\| \\
        & \quad + \|\nabla_\theta \nabla_u J({z^\star}')\| \|[\nabla^2_u J(z^\star)]^{-1} -[\nabla^2_u J({z^\star}')]^{-1}\| \\
        & \leq \frac{1}{\mu} L_2 \|z^\star - {z^\star}'\| + L_1 \|[\nabla^2_u J(z^\star)]^{-1} -[\nabla^2_u J({z^\star}')]^{-1}\|.
    \end{aligned}
    \end{equation}
    Further, using the Lipschitz property of $\nabla_u^2 J(z^\star)$ and boundness of $\|\nabla_u^2 J(z^\star)\|$, we have
    \begin{equation}\label{eq:l2}
    \begin{aligned}
         \|[\nabla^2_u J(z^\star)]^{-1} -[\nabla^2_u J({z^\star}')]^{-1}\| & \leq \|[\nabla_u^2 J(z^\star)]^{-1}\|\|\nabla_u^2 J({z^\star}') - \nabla_u^2 J(z^\star)\|\|[\nabla_u^2 J({z^\star}')]^{-1}\| \\   
         & \leq \frac{L_3}{\mu^2} \|z^\star - {z^\star}'\|.
    \end{aligned}
    \end{equation}
    We note 
    \begin{equation}\label{eq:l3}
        \|z^\star-{z^\star}'\| \leq \|\theta - \theta'\| + \|u^\star(\theta)-u^\star(\theta')\| \leq \frac{\mu+L_1}{\mu} \| \theta-\theta'\| 
    \end{equation}
    Combining equation~\eqref{eq:l1},\eqref{eq:l2} and~\eqref{eq:l3} yields
    \begin{equation}\label{eq:lip_ug}
        \|\nabla_\theta u^\star(\theta) - \nabla_\theta u^\star (\theta') \| \leq  (\frac{L_2}{\mu} + \frac{L_1L_2 + L_1L_3}{\mu^2} + \frac{L_1^2L_3}{\mu^3}) \|\theta - \theta'\|.
    \end{equation}
    From the bound in~\eqref{eq:lip_ug}, we obtain
\begin{align*}
    \left\| \nabla^2_{\theta} u^\star(\theta) \right\| \leq \frac{L_2}{\mu} + \frac{L_1 L_2 + L_1 L_3}{\mu^2} + \frac{L_1^2 L_3}{\mu^3}.
\end{align*}
As defined in Proposition~\ref{prop:bounded_sensitivity}, the constants \( M_{u1} \) and \( M_{u2} \) bound the first- and second-order derivatives of the policy. Specifically, for all planning steps \( i \),
\[
\left\| \nabla_\theta u_i^\star(x; \theta) \right\| \leq M_{u1}, \quad
\left\| \nabla^2_\theta u_i^\star(x; \theta) \right\| \leq M_{u2}.
\]
\end{proof}

\section{Experiment Details}\label{ap:exp}

\subsection{System/Environment Setup}
\noindent
\textbf{Nonlinear Control Tasks.} We follow the experimental setup from~\cite{jin2020pontryagin} and train approaches, including DiffOP and RL-based MPC methods, using their provided environment implementations.

\noindent
\textbf{Voltage Control with Power Injection Constraints. } 
We follow the experimental setup from~\cite{feng2023bridging}, using the IEEE 13-bus distribution system. The nominal voltage magnitude at each bus (excluding the substation) is 4.16 kV, with an acceptable operating range of ±5\%, i.e., [3.952 kV, 4.368 kV].  
All experiments are conducted using Pandapower~\cite{turitsyn2011options} as the nonlinear power flow simulator for performance evaluation.
The system state is defined as \( x = [v, q, v^0] \in \mathbb{R}^9 \), where \( v \in \mathbb{R}^3 \) denotes voltage magnitudes, \( q \in \mathbb{R}^3 \) is the reactive power injection, and \( v^0 \) is the initial voltage.  
The control action \( u \in \mathbb{R}^3 \) represents the change in power injection.  
We emphasize that system state constraints are explicitly enforced within the optimization-based control policy.

\noindent
\textbf{Dynamics discretization.} The dynamical systems for nonlinear system control are discretized using the Euler method: $x_{t+1} = x_t + \Delta t \cdot f(x_t, u_t)$ with the discretization interval $\Delta t = 0.05$s or $\Delta t = 0.1$s.
For voltage control task, $\Delta t = 1$ [seconds].

\subsection{Implementations for Control Policies. }

\textbf{Initialization.} 
For the \textit{nonlinear control} tasks, all optimization-based policies are initialized following the same procedure used in~\cite{jin2020pontryagin}, while their original setup was in a supervised imitation learning context.
For the \textit{voltage control} tasks, all approaches are initialized with $A = 0.5I$ and $C_v = I$. Under this initialization, the resulting policy does not drive the voltage within the acceptable 5\% deviation range.

\noindent
\textbf{DiffOP.} 
For the policy gradient estimation in both nonlinear control and voltage control, we choose $N=10$ as the number of trajectory samples, $\beta = 0.01$ as the exploration rate, $\eta = 0.1$ to be the learning rate. DiffOP (Step) uses a planning horizon of $H=6$ for the Cartpole and Voltage control tasks, and $H=12$ for the Robot Arm and Quadrotor control tasks. DiffOP (Traj) sets the horizon to $H=T$, equal to the full episode length of each environment.

\noindent
\textbf{Pontryagin Differentiable Programming (PDP) (offline).}
We follow the methodology outlined~\cite{jin2020pontryagin} and utilize their published code to conduct our experiments. 
For the voltage control task, we generate 10,000 trajectories using expert demonstrations from TASRL~\cite{feng2023bridging}, and train a policy using the PDP framework. The planning horizon is $H=T$ as DiffOP(traj). 

\noindent
\textbf{RLMPC-TD.} A model-free reinforcement learning approach that updates a short-horizon MPC policy using temporal-difference (TD) Q-learning. 
We implement this approach using the MPC-RL framework~\cite{reinhardt2024mpc4rl}, and for each control task, we perform a grid search over learning rates $\{10^{-1}, 10^{-2}, \ldots, 10^{-8}\}$ to select the best-performing setting.

\noindent
\textbf{RLMPC-DPG.} 
A model-free reinforcement learning approach that updates a short-horizon MPC policy using both temporal-difference (TD) Q-learning and deterministic policy gradient (DPG).
We implement this approach using the MPC-RL framework~\cite{reinhardt2024mpc4rl}, and for each control task, we perform a grid search over learning rates $\{10^{-1}, 10^{-2}, \ldots, 10^{-8}\}$ and explore rate $\{10^{-1}, 10^{-2}, \ldots, 10^{-5}\}$ to select the best-performing setting.

\subsection{Running time}
Table~\ref{tab:voltage_time} reports the average runtime for voltage control task. All experiments were performed on a single thread of an Intel Core i7-11700K CPU to ensure a fair comparison, as previous works~\cite{jin2020pontryagin, xu2024revisiting} also compare running times using CPUs. 
We note that DiffOP(Step) uses a shorter horizon compared to DiffOP(Traj), which leads to less forward and backward computation per optimization. However, within each trajectory, DiffOP(Traj) solves the optimal control problem and computes gradients only once, while DiffOP(Step) performs both operations at every time step, resulting in increased total training time.
Compared to PDP (offline), we observe that the policy gradient computation in DiffOP does not significantly increase the backward pass time. 
PDP (offline) is trained entirely on offline data, with the total training time determined by the number of training iterations. In our setup, we use 10,000 iterations, evaluating the gradients using 10 trajectories per iteration.
Since RLMPC-TD and RLMPC-DPG are implemented using the packaged MPC4RL framework~\cite{reinhardt2024mpc4rl}, we report their total runtime without separating forward and backward passes.

\begin{table}[t]
\centering
\begin{tabular}{p{6cm}ccccc} \hline
         & \multicolumn{2}{c}{DiffOP} &  \multirow{2}{*}{PDP(offline)} & \multirow{2}{*}{RLMPC-DPG} & \multirow{2}{*}{RLMPC-TD }  \\
         & Traj      & Step     &                          &                         \\ \hline
Forward (Solving the policy) (seconds) &     0.026        &   0.011     &  0.026    &                   -      &   -                       \\
Backward (Calculate gradients) (seconds) &    0.010         &  0.002      &    0.010  &              -            &       -                  \\
\hline 
Total training time (hours)   &   0.37          &   1.36      &   1.94   &               0.86         &          0.72             \\\hline  
\end{tabular}
\caption{The average runtime for voltage control task.}\label{tab:voltage_time}
\end{table}

\end{document}